\documentclass{article}
\usepackage{authblk}

\usepackage[T1]{fontenc}

\usepackage{graphicx}
\usepackage{color}

\usepackage{amsmath,amssymb,units,stmaryrd,amsthm,xcolor}
\usepackage[hidelinks]{hyperref}
\usepackage{yfonts,array}
\usepackage[hidelinks]{hyperref}
\usepackage[nameinlink]{cleveref}
\usepackage{upgreek,xcolor,bm,tikz,stackrel}
\usepackage[inline, shortlabels]{enumitem}
\usepackage{graphicx}
\usepackage[all]{xy}
\usepackage{tikz-cd}
\usetikzlibrary{decorations.pathmorphing}
\usepackage{wasysym}
\usetikzlibrary{arrows}
\usepackage[british]{babel}
\usepackage{multicol}
\usepackage{mathtools}

\usepackage[linesnumbered,ruled]{algorithm2e}
\usepackage{afterpage}
\SetKwComment{Comment}{// }{ }
\SetCommentSty{textrm}

\newcommand\ignore[1]{}

\newcommand{\define}[1]{\textbf{#1}}

\newcommand{\lm}{\vec{\mathcal  Q}}
\newcommand{\domlm}{|\vec {\mathcal  Q}|}

\newcommand{\iiff}{\mathbin \Leftrightarrow}
\newcommand{\imp}{\mathbin \Rightarrow}
\newcommand{\dimp}{\mathbin \Leftarrow}
\newcommand{\dneg}{{\sim}}

\newcommand{\Om}{\Omega}
\newcommand{\om}{\omega}

\newcommand{\gtl}{\mathsf{GTL}}
\newcommand{\gtlrel}{\gtl_\mathrm{rel}}
\newcommand{\gtlreal}{\gtl_\mathbb R}
\newcommand{\ltl}{\mathsf{LTL}}

\newcommand{\lanfull}{{\mathcal L}}

\newcommand{\lgt}[1]{|#1|}
\newcommand{\nx}{\mathsf X}
\newcommand{\y}{\mathsf Y}
\newcommand{\nec}{\mathsf G}
\newcommand{\has}{\mathsf H}
\newcommand\until{\mathbin \mathsf  U}
\newcommand\since{\mathbin \mathsf  S}
\newcommand\release{\mathbin \mathsf R}

\newcommand{\ps}{\mathsf F}
\newcommand{\past}{\mathsf P}

\newcommand{\val}[1]{\lb #1 \rb}
\newcommand{\moments}[1]{\mathbb M_{ #1 }}
\newcommand{\type}[1]{\mathbb T_{ #1 }}
\newcommand{\peq}{\leq}

\newcommand\lb{\left\llbracket}
\newcommand\rb{\right\rrbracket}
\newcommand\<{\left (}
\renewcommand\>{\right )}
\newcommand\cbra{\left \{}
\newcommand\cket{\right \}}
\DeclareSymbolFont{AMSb}{U}{msb}{m}{n}

\newcommand{\compa}{\lesseqgtr}
\newcommand{\ptype}{\type\infty}
\newcommand\eqdef{\mathrel{\mathop:}=}
\newcommand{\circop}{\nx^{-1}}
\newcommand{\yop}{\y^{-1}}
\newcommand{\CMod}{\mathfrak C}
\newcommand{\CIcon }{\CMod}
\newcommand\cqm[1]{\nicefrac{\CMod}{#1}}
\newcommand\powerset{\mathcal P}

\newtheorem{theorem}{Theorem}[section]
\newtheorem{corollary}[theorem]{Corollary}
\newtheorem{lemma}[theorem]{Lemma}
\newtheorem{proposition}[theorem]{Proposition}

\theoremstyle{definition}
\newtheorem{definition}[theorem]{Definition}
\newtheorem{example}[theorem]{Example}

\theoremstyle{remark}
\newtheorem{remark}[theorem]{Remark}

\usepackage{hyperref}
\title{G\"odel--Dummett linear temporal logic}

\author[1,2]{Juan Pablo Aguilera\footnote{\href{mailto:aguilera@logic.at}{\tt aguilera@logic.at}}}
\author[3]{Mart\'in Di\'eguez\footnote{\href{mailto:martin.dieguezlodeiro@univ-angers.fr}{\tt martin.dieguezlodeiro@univ-angers.fr}}}
\author[1,4]{David Fern\'andez-Duque\footnote{\href{mailto:fernandez-duque@ub.edu}{\tt fernandez-duque@ub.edu}}}
\author[1]{Brett McLean\footnote{\href{mailto:brett.mclean@ugent.be}{\tt brett.mclean@ugent.be}}}

\affil[1]{Department of Mathematics WE16, Ghent University,Ghent, Belgium}
\affil[2]{Institute of Discrete Mathematics and Geometry, Vienna University of Technology, Vienna, Austria}
\affil[3]{LERIA, University of Angers, Angers, France}
\affil[4]{Department of Philosophy, University of Barcelona, Barcelona, Spain}

\begin{document}




%
%
%
%
%

\maketitle

\begin{abstract}We investigate a version of linear temporal logic whose propositional fragment is G\"odel--Dummett logic (which is well known both as a superintuitionistic logic and a t-norm fuzzy logic). We define the logic using two natural semantics: first a real-valued semantics, where statements have a degree of truth in the real unit interval and second a `bi-relational' semantics. We then show that these two semantics indeed define one and the same logic: the statements that are valid for the real-valued semantics are the same as those that are valid for the bi-relational semantics. This G\"odel temporal logic does not have any form of the finite model property for these two semantics: there are non-valid statements that can only be falsified on an infinite model. However, by using the technical notion of a quasimodel, we show that every falsifiable statement is falsifiable on a finite quasimodel, yielding an algorithm for deciding if a statement is valid or not. Later, we strengthen this decidability result by giving an algorithm that uses only a polynomial amount of memory, proving that G\"odel temporal logic is PSPACE-complete. We also provide a deductive calculus for G\"odel temporal logic, and show this calculus to be sound and complete for the above-mentioned semantics, so that all (and only) the valid statements can be proved with this calculus.  
\end{abstract}



%
%
%



\section{Introduction}

The importance of temporal logics and, independently, of fuzzy logics in computing is well established \cite{degola16a, 364485}. 
The potential usefulness of their combination is clear: for instance, it would provide a natural framework, in symbolic artificial intelligence, for the specification and verification of systems dealing with vague or incomplete data~\cite{Kruse1991}.

One of the most thoroughly studied fuzzy logics is \emph{G\"odel logic} (also called G\"odel--Dummett logic) \cite{10.1007/978-3-642-16242-8_4,BPZ07}.  G\"odel logic was introduced by G\"odel \cite{Godel} in his proof that intuitionistic logic is not finite-valued and later axiomatized by Dummett \cite{10.2307/2964753}. It is an extremely useful and natural framework because it is both a t-norm fuzzy logic and an extension of intuitionistic logic, thus capturing both the former's approach to reasoning with vagueness as well as the latter's approach to evidence-based reasoning. G\"odel logic has been considered as a foundation for logic programming~\cite{AlsinetCGSS08,duboisLP91}, answer set programming~\cite{pearce96,Pearce06}, and parallel $\uplambda$-calculus~\cite{Avron91,aschieriCG17}. 

Applications of fuzzy temporal reasoning in computer science, engineering, and artificial intelligence are numerous~\cite{DuboisP89,lamineK00,MUKHERJEE20131452, GERBER2008351}.
One very high-profile and promising current focus of artificial intelligence research is
\emph{autonomous driving}, where fuzzy logic is often the underlying framework for the controllers used to operate vehicles~\cite{4515884,Shukla}. The use of fuzzy logic allows for the specification of control rules in natural language that can be written and understood by non-specialists. Rules applied to autonomous vehicles include if--then rules that are triggered under certain preconditions. For instance, these rules could include the following:
\begin{enumerate}[label=(\alph*)]
	\item\label{statement1} If the vehicle ahead is \emph{close} and the car's speed is \emph{high}, then the car must \emph{decelerate}.  
\end{enumerate}	
 Usually the type of rules declared are reactive, i.e.~the trigger depends only on the current state, and the action is implemented at the immediately following instant. 
Expressing rules whose evaluation implicates a possible infinite number of time instants is very complicated to do with existing approaches. Thus a more powerful framework is needed.

Perhaps the best-known and most successful formalism for temporal reasoning is \textit{linear-time temporal logic} (or propositional linear-temporal logic; $\ltl$)~\cite{pnueli}: the extension of propositional logic with a variety of temporal modalities such as `eventually', `henceforth', and `until'. Indeed the success of $\ltl$ in program and systems verification resulted in the 1996 Turing Award being conferred on Pnueli.

Several variants of $\ltl$ lacking the classical, binary conception of truth have been investigated, with a variety of application domains. In logic programming and non-monotonic reasoning, equilibrium logic~\cite{Pearce06} was given an $\ltl$ extension in~\cite{taspa,taspb}, based on the G\"odel logic with three truth values $\mathsf G_3$ (also known as the logic of here-and-there~\cite{Hey30}). Such an $\ltl$ extension of $\mathsf G_3$ is axiomatised in~\cite{BalbianiDieguezJelia}. 
Combinations of $\ltl$ with intuitionistic propositional logic have been considered as a framework for reasoning about dynamical systems \cite{F-D18}.
 Combinations of fuzzy logic with $\ltl$ have been used to enable vague time references, for example in~\cite{327373,10.1145/2629606}. 
Those approaches have many fields of application: disaster management, robot control, or smart home systems; that is, to dynamic systems whose evolution depends on uncertain conditions~\cite{Lu2010}. 

Motivated by the many applications of G\"odel logic and its temporal extensions, particularly $\ltl$ extensions, in this paper we provide a detailed investigation of a linear-temporal extension of the real-valued G\"odel logic (where degrees of truth can take any value in the unit interval $[0,1]$). 
Specifically, we investigate a temporal extension of propositional G\"odel logic whose syntax can express all the familiar $\ltl$ \emph{future} and \emph{past} modalities, `until', `henceforth', `since', and so on. We call this logic $\gtl$ (short for \emph{G\"odel temporal logic}). This investigation extends that of \cite{10.1007/978-3-031-15298-6_2,gtlkr}, whose results apply to the language with only the \emph{unary}, \emph{future} modalities `eventually' and `henceforth'.

This language allows us to express statements whose satisfaction may require arbitrarily large amounts of time. For instance, we can express rules like: 
	\begin{enumerate}[label=(\alph*), start=2]
		\item\label{tstatement1} If the petrol tank is \emph{low}, it will remain low \emph{until} either the tank is refilled or it becomes empty.
		\item\label{tstatement2} If \emph{rockfall} is observed near Mo\^{u}tiers, then \emph{henceforth} drive with caution  in the vicinity of Mo\^{u}tiers.
	\end{enumerate}
	Unlike example \ref{statement1}, expressing \ref{tstatement1} and \ref{tstatement2} requires a language with both a vague component and an elaborate temporal component capable of referring to actions potentially in the distant future.


In fact, $\gtl$ possesses two natural semantics, corresponding to whether it is viewed as a fuzzy logic or a superintuitionistic logic.
As a fuzzy logic, propositions take values in $[0,1]$, and truth values of compound propositions are defined using standard operations on the reals.
As a superintuitionistic logic, models consist of \emph{bi-relational structures}: a set equipped with a partial order to interpret implication as on an intuitionistic Kripke frame and a function to interpret the $\ltl$ tenses.
 However, notably, the two semantics give rise to the same set of valid formulas (\Cref{thm:equal}).

One may be concerned that the passage from two-valued to infinite-valued $\ltl$ could lead to an explosion in computational complexity, as has been known to happen when combining sub-classical logics with transitive modal logics: for these combinations it is often algorithmically undecidable if a formula is valid~\cite{Vidal21}, or decidable with only a superexponential space upper bound being known \cite{BalbianiDF21}. 
As we will see, this need not be the case: our combination of G\"odel--Dummett logic with linear temporal logic remains \textsc{pspace}-complete (\Cref{thmPSPACE}), the minimal possible complexity given that classical $\ltl$ embeds into it.
This is true even when the syntax is enriched with the dual  implication connective~\cite{CecyliaRauszer1980}, which has been argued in \cite{BilkovaFK21} to be useful for reasoning with incomplete or inconsistent information. In contrast, the $\ltl$ extension of equilibrium logic mentioned earlier is \textsc{expspace}-complete~\cite{10.1109/LICS.2015.65}.

Even the decidability of $\gtl$ is quite surprising, since its `modal companion' $\mathsf{S4.3} \times \ltl$ is not recursively axiomatisable~\cite{10.1093/logcom/11.6.909}, and, as we show, $\gtl$ does not enjoy the finite model property for either of its two semantics (i.e.~there are non-valid formulas whose falsifiability can only be seen using an infinite model).

For G\"odel logics, it is often possible to prove decidability despite the lack of finite model property by considering alternative semantics (see e.g.~\cite{CaicedoMRR17}).
For example, the logic $\mathsf{GS4}$ does not have the real-valued finite model property, but it does have the bi-relational finite model property~\cite{BalbianiDF21}.
Since $\gtl$ does not enjoy \emph{either} version of the finite model property, we instead introduce \emph{quasimodels}, which do satisfy their own version of the finite model property.
Quasimodels are not `true' models, because the relation indicating the `next' instant of time is no longer necessarily a function, but quasimodels give rise to standard bi-relational models by unwinding.
Similar structures were used to prove upper complexity bounds for dynamic topological logic \cite{Fernandez09,dtlaxiom} and intuitionistic temporal logic \cite{F-D18}, but they are particularly effective in the setting of G\"odel temporal logic, yielding the optimal \textsc{pspace} upper bound.

A natural deductive calculus for $\gtl$ can readily be defined by combining axioms for G\"odel logic (with dual implication) together with axioms for linear temporal logic. However, as in the classical setting, a standard attempt to prove the completeness of these axioms via a \emph{canonical model} will not yield a structure where the tenses are interpreted correctly. For $\ltl$, one may use the technique of \emph{filtration} to obtain a finite model and then use the method of associating a \emph{characteristic formula} to each world to characterise it up to bisimulation, using these to reason that infinitary tenses (such as `henceforth') are indeed interpreted correctly~(see e.g.~\cite{DemriGLBook}).
For $\gtl$, we are able to prove completeness by using quasimodels to provide an analogue of the finite filtrated model (\Cref{theocomp}), with the additional complication that in this sub-classical setting it is necessary to assign \emph{two} characteristic formulas to each world.



\subsubsection*{Related work}\smallskip

As already mentioned, many different extensions of temporal reasoning beyond the classical Boolean propositional framework have been studied, for many different reasons. Here we briefly describe the work most closely related to our own, and highlight the relationships with our approach and results.

In \cite{LPAR-21:Godel_logics_and_fully}, Baaz and Preining consider a fragment of monadic \emph{first-order} $\ltl$. By a reduction from monadic first-order G\"odel logic, they show that the valid sentences of this monadic first-order $\ltl$ are not computably enumerable. This makes our own result all the more surprising: in the classical setting, monadic first-order (indeed, second-order) logic is already decidable, from which decidability for $\ltl$ follows, although a \textsc{pspace} complexity bound requires additional proof. In the G\"odel setting, we lose the decidability of monadic first-order logic, suggesting that G\"odel $\ltl$ should at the very least also suffer in terms of complexity. Yet this is not the case, and the \textsc{pspace} upper bound is preserved.

\emph{Temporal equilibrium logic} \cite{taspb}  is a non-monotonic $\ltl$ that is one of the most popular semantics for temporal answer set programming.
Its definition consists of two main components:
a monotonic basis in terms of an $\ltl$ extension of the G\"odel logic $\mathsf G_3$ and a selection criterion to select minimal models.
Such a selection criterion can be seen as a second-order quantification over the set of temporal $\mathsf G_3$ models.
As a result, the complexity of the satisfiability problem of temporal equilibrium logic turns out to be \textsc{expspace} \cite{10.1109/LICS.2015.65}.
There exists, nevertheless, an interesting fragment, called splittable temporal logic programs \cite{10.1007/978-3-642-20895-9_9}, whose complexity remains open. The approach presented in this paper could help to determine the complexity of the satisfiability problem for the case of such programs.

Finally, there have been several investigations into fuzzy \emph{interval temporal logics}, for example \cite{BADALONI2006872, SCHOCKAERT20081158}. This work focuses on the study of fuzzy time periods and reasoning about them (e.g. ``$X$ ends before $Y$ begins''). They do not work with a non-classical implication, but rather focus on disjunctions of atomic expressions. Accordingly, the complexity of the tasks they consider is \textsc{np}. In principle, the relations they consider can be represented in a language such as ours, although they work in continuous time. Indeed, a continuous-time analogue of our framework would be a natural avenue for future exploration.

\subsubsection*{Structure of the article}\smallskip
 \noindent\Cref{SecBasic}: \textbf{Syntax and semantics}. We introduce the temporal language that we work with, and then introduce both the real-valued semantics and bi-relational semantics for G\"odel temporal logic.

\smallskip

\noindent \Cref{sec:real}: \textbf{Real-valued versus~bi-relational validity}.  We prove  the equivalence of these two semantics, that is, that they yield the same validities (\Cref{thm:equal}). 

\smallskip

\noindent \Cref{SecNDQ}: \textbf{Labelled systems and quasimodels}. We first note that we do not have a finite model property for either of these semantics. But then we define quasimodels (\Cref{compatible}), and in later sections show that our G\"odel temporal logic is sound and complete for the class of finite quasimodels. 

\smallskip

\noindent \Cref{SecGen}: \textbf{From quasimodels to bi-relational models}.  We show that G\"odel temporal logic is sound for (all) quasimodels, constructing a bi-relational model from an arbitrary quasimodel by unwinding selected paths within the quasimodel.

\smallskip

\noindent \Cref{Sec:quotient}: \textbf{From bi-relational models to finite quasimodels}. Given a bi-relational model falsifying a formula, we use a quotient construction to produce a finite (exponential in the length of the formula) quasimodel also falsifying the formula. This completes the proof that the semantics of G\"odel temporal logic can be reduced to finite quasimodels, yielding the decidability of G\"odel temporal logic (\Cref{theorem:decide}).

\smallskip

\noindent \Cref{secAx}: \textbf{A Hilbert-style deductive calculus}. Introduces our Hilbert calculus for the valid formulas of G\"odel temporal logic.

\smallskip

\noindent \Cref{secCan}: \textbf{The canonical system and canonical quasimodel}. Introduces the canonical system, essentially a standard canonical model familiar from modal logic.
However, this system is not a `true' model, so we proceed to construct the canonical quasimodel by re-using the quotient construction from \Cref{Sec:quotient}.

\smallskip

\noindent \Cref{SecChar}: \textbf{Characteristic formulas}. For the purpose of proving that the canonical quasi\-model is indeed a quasimodel, we define characteristic formulas, which define models locally up to bisimulation.

\smallskip

\noindent \Cref{SecComp}: \textbf{Completeness}. Characterisistic formulas are used to prove that the canonical quasimodel satisfies all required properties of a quasimodel, from which completeness of our axiomatisation follows (\Cref{theocomp}).

\smallskip

\noindent \Cref{Sec:PSPACE}: \textbf{PSPACE-complete complexity}. We refine our decidability result, showing that G\"odel temporal logic is in fact \textsc{pspace}-complete (\Cref{thmPSPACE}).
\smallskip

\noindent \Cref{SecConc}: \textbf{Concluding remarks}. We summarise our results, contrast them with the properties of similar logics, and suggest some directions for future work.

\smallskip

There are several routes through the paper for readers interested only in a portion of the results.
Section~\ref{SecBasic} is common background needed for the rest of the paper.
Section~\ref{sec:real} may be skipped by readers not interested in real-valued semantics, or alternatively read after the rest of the paper.
Sections~\ref{SecGen}--\ref{Sec:quotient} provide a full proof of decidability, and readers interested only in complexity may then skip to Section~\ref{Sec:PSPACE}.
Readers interested in completeness of the Hilbert calculus should read Sections~\ref{SecGen}--\ref{SecComp}; note that the proof of completeness builds on the proof of decidability.

\section{Syntax and semantics}\label{SecBasic}

In this section we first introduce the temporal language we work with and then two possible semantics for this language: real-valued semantics and bi-relational semantics. 

Fix a countably infinite set $\mathbb P$ of propositional variables. Then the \define{G\"odel temporal language} $\lanfull$ is defined by the grammar (in Backus--Naur form):
\[\varphi,\psi :=   p  \ |\ \varphi\wedge\psi \ |\ \varphi\vee\psi  \ |\ \varphi\imp \psi  \ |\  \varphi\dimp \psi  \ |  \ \nx\varphi \ |\ \y \varphi \  |  \ \nec \varphi \ |\ \has\varphi \  | \  \varphi \until \psi \ |\ \varphi \since \psi , \]
where $p\in \mathbb P$. Here, $\nx$ is read as `ne\textbf xt', $\y$ as `\textbf yesterday, $\nec$ as `\textbf going (to always be)', $\has$ as `\textbf has (always been)', $\until$ as `\textbf until', and $\since$ as `\textbf since'. The connective $\dimp$ is \emph{co-implication} and represents the operator that is dual to implication \cite{Wolter1998}.

 We also use the following abbreviations: 
 \begin{itemize}
 \item
 $\top$ abbreviates $p\imp p$, for some fixed, but unspecified, $p \in \mathbb P$;
 \item
 $\bot$ abbreviates $p\dimp p$;
 \item
 $\neg\varphi$ abbreviates $\varphi\imp \bot$;
 \item
 $\varphi\iiff \psi$ abbreviates $(\varphi \imp \psi) \wedge (\psi \imp \varphi)$ (not the formula $(\varphi \imp \psi) \wedge (\varphi \dimp \psi)$);
 \item $\ps\varphi$, read as `\textbf future', abbreviates $\top \until {\varphi}$;
 \item $\past\varphi$, read as `\textbf past', abbreviates $\top \since {\varphi}$.
 \end{itemize}
Although we will not need them, note that we can also define `\textbf release' as $\varphi \release \psi \coloneqq (\psi \until \varphi) \vee \nec \psi$ and the past analog of release similarly. Often the alternative names `eventually' and `henceforth' are used for `future' and `going' respectively.

The rules presented in the introduction correspond to the following temporal formulas:

\begin{description}
	\item[\ref{statement1}]   $\nec(\mathtt{aheadClose} \wedge \mathtt{speedHigh} \imp \mathtt{decelerate}) $;
	\item[\ref{tstatement1}]  $\nec(\mathtt{tankLow} \imp (\mathtt{tankLow} \until (\mathtt{refill} \vee \mathtt{tankEmpty})))$;
	\item[\ref{tstatement2}]  $\nec(\mathtt{rockfall} \wedge \mathtt{moutiers} \imp \nec(\mathtt{moutiers} \imp \mathtt{caution}))$.
\end{description}

We now introduce the first of our semantics for the G\"odel temporal language: \emph{real-valued semantics}, which views $\lanfull$ as a \emph{fuzzy logic} (enriched with temporal modalities). In the definition, $[0,1]$ denotes the real unit interval.

\begin{definition}[real-valued semantics]\label{DefRSem}
A \define{flow} is a pair $\mathcal T = (T,S)$,
where $ T $ is a set and $S \colon T \to T$ is a bijection.
A \define{real valuation} on $\mathcal T$ is a function $V \colon \lanfull \times T \to  [0,1]$ such that, for all $t\in T$, the following equalities hold.
\begin{center}
\begin{tabular}{rclrcl}
$V(\varphi\wedge\psi,t) $&$=$&$\min \{ V(\varphi ,t), V( \psi,t) \}$
&
$V(\varphi \vee \psi,t) $&$=$& $\max \{ V(\varphi ,t) , V( \psi,t) \}$\\\\
$V(\varphi \imp \psi,t) $&$=$& $
\begin{cases}
1\text{ if }V(\varphi, t){\leq} V(\psi, t)\\
V(\psi,t )\text{ otherwise}
\end{cases}$
&
$V(\varphi \dimp \psi,t) $&$=$&$
\begin{cases}
0\text{ if }V(\varphi, t){\leq} V(\psi, t)\\
V(\varphi,t )\text{ otherwise}
\end{cases}$\\\\
$V(\nx\varphi, t)$&$=$&$ V( \varphi,S(t))$&
$V(\y\varphi, t)$&$=$&$ V( \varphi,S^{-1}(t))$\\\\
$V(\nec \varphi,t) $&$=$&$ \inf_{n<\infty}   V(\varphi,S^n(t))$
&
$V(\has \varphi,t) $&$=$&$ \inf_{n<\infty}   V(\varphi,S^{-n}(t))$
\\\\
\multicolumn{6}{c}{
$V(\varphi \until \psi,t) = \sup_{n<\infty}  \min\{ V(\varphi,S^0(t)) , \dots, V(\varphi,S^{n-1}(t)),  V(\psi,S^n(t))\} $}\\\\
\multicolumn{6}{c}{
$V(\varphi \since \psi,t) = \sup_{n<\infty}  \min\{ V(\varphi,S^0(t)) , \dots, V(\varphi,S^{-(n-1)}(t)),  V(\psi,S^{-n}(t))\} $}\\
\end{tabular}
\end{center}
A flow $\mathcal  T$ equipped with a valuation $V$ is a \define{real-valued (G\"odel temporal) model}.
\end{definition}
 Informally, the real number $V(\varphi, t)$ records the degree of truth of $\varphi$ at time $t$.

The second semantics, \emph{bi-relational} semantics, views $\lanfull$ as a (consistent) \emph{superintuitionistic logic}, temporally enriched.
Below, define $\vec S(w,t) = (w,S(t))$, and if $X$ is a subset of the domain of $\vec S$, then $\vec SX = \{\vec S(x) \mid x \in X\}$.

\begin{definition}[bi-relational semantics]\label{DefKSem}
A \define{(G\"odel temporal) bi-relational frame} is a quadruple $\mathcal  F=(W,T,{\leq},S)$ where $(W,{\leq})$ is a linearly ordered set and $(T,S)$ is a flow.
A \define{bi-relational valuation} on $\mathcal  F$ is a function $\lb\,\cdot\,\rb\colon\lanfull \to 2^{W\times T}$ such that, for each $p \in \mathbb P$, the set $\lb p \rb$ is \emph{downward closed} in its first coordinate, and the following equalities hold.

\begin{center}
\begin{tabular}{rclrcl}
 $\lb\varphi\wedge\psi\rb$&$=$&$\lb\varphi\rb\cap \lb\psi\rb$&
 $\lb\varphi\vee\psi\rb $&$=$&$ \lb\varphi\rb\cup \lb\psi\rb$\\\\
\multicolumn{6}{l}{$\lb\varphi\imp\psi\rb = \{ (w,t) \in W\times T \mid  \forall v \leq  w\ ((v,t) \in \lb\varphi \rb  
\text{ implies } (v,t)\in \lb\psi \rb ) \}$}\\\\
\multicolumn{6}{l}{$\lb\varphi\dimp\psi\rb  =  \{ (w,t) \in W\times T \mid \exists v \geq  w\ ((v,t) \in \lb\varphi \rb  
\text{ and } (v,t)\notin \lb\psi \rb ) \}$} \\\\
$\val{\nx\varphi}  $&$=$&$   \vec S^{-1} \val\varphi$
&
$\val{\y\varphi}  $&$=$&$   \vec S \val\varphi$
\\\\
$\val{\nec \varphi} $&$=$&$  \bigcap_{n<\infty} \vec S^{-n}\val\varphi $   &
$\val{\has \varphi} $&$=$&$  \bigcap_{n<\infty} \vec S^{n}\val\varphi $   \\\\
\multicolumn{6}{l}{$\val{\varphi \until \psi} = \bigcup_{n<\infty} ( \vec S^{0} \val\varphi \cap   \dots \cap  \vec S^{-(n-1)} \val\varphi \cap \vec S^{-n} \val\psi)$}
\\\\
\multicolumn{6}{l}{$\val{\varphi \since \psi} = \bigcup_{n<\infty} ( \vec S^{0}   \val\varphi \cap   \dots \cap  \vec S^{n-1} \val\varphi \cap \vec S^{n} \val\psi)$}
\\
\end{tabular}
\end{center}
A bi-relational frame $\mathcal  F$ equipped with a valuation $\lb\,\cdot\,\rb$ is a \define{(G\"odel temporal) bi-relational model}.
\end{definition}

Informally, the set $\{w \in W \mid (w, t) \in \val \varphi\}$ records the degree of truth of $\varphi$ at time $t$. This semantics combines standard intuitionistic Kripke frame semantics for the implications based on the order $\leq$ (read downward
) and standard semantics for the tenses based on $S$: for example, $(w,t)\in \val{\varphi \until \psi}$ if and only if there exists $n\geq 0$ such that $(w, S^i(t)) \in \val\varphi$ for all $0 \leq i < n$, and $(w,S^n(t)) \in \val\psi$.
Note that by structural induction, the valuation of \emph{any} formula $\varphi$ in $\lanfull$ is downward closed in its first coordinate, in the sense that if $(w,t)\in \val\varphi$ and $v\leq w$, then $(v,t)\in \val \varphi$.

For both real-valued semantics and bi-relational semantics, validity of $\lanfull$-formulas is defined in the usual way.

\begin{definition}[validity]
Given a real-valued model $\mathcal  X = (T, S, V)$ and a formula $\varphi\in \lanfull$, we say that $\varphi$ is \define{globally true} on $\mathcal  X$, written $\mathcal  X\models\varphi$, if for all $t \in T$ we have $V(\varphi, t) = 1$.
Given a bi-relational model $\mathcal  X = (\mathcal F, \lb\,\cdot\,\rb)$ and a formula $\varphi\in \lanfull$, we say that $\varphi$ is \define{globally true} on $\mathcal  X$, written $\mathcal  X\models\varphi$, if $\val\varphi =W \times T$. 

If $\mathcal  X$ is a flow or a bi-relational frame, we write $\mathcal  X\models\varphi$  and say $\varphi$ is \define{valid} on $\mathcal X$, if $\varphi$ is globally true for every valuation on $\mathcal  X$. If $\Om$ is a class of flows, frames, or models, we say that $\varphi\in\lanfull$ is \define{valid} on $\Om$ if, for every $\mathcal  X\in \Omega$, we have $\mathcal  X\models\varphi$. If $\varphi$ is not valid on $\Om$, it is \define{falsifiable} on $\Om$.
\end{definition} 

We define:
\begin{itemize}
\item
 the logic $\gtlreal$ to be the set of  $\lanfull$-formulas that are valid over the class of all \emph{flows};
 \item
 the logic $\gtlrel$ to be the set of $\lanfull$-formulas that are valid over the class of all \emph{bi-relational frames}.
 \end{itemize}

\section{Real-valued versus~bi-relational validity}\label{sec:real}

In this section, we show that an arbitrary $\lanfull$-formula is real valid if and only if it is bi-relationally valid. That is, $\gtlreal = \gtlrel$. The two directions of this equivalence are \Cref{LemmaRealToKripke} and \Cref{LemmaKripkeToReal} below.

\begin{lemma}\label{LemmaRealToKripke}
Suppose that $\varphi$ is an $\lanfull$-formula that is not real valid. Then $\varphi$ is not bi-relationally valid.
\end{lemma}
\begin{proof}
Let $(T, S)$ be a flow, $V$ a real valuation on $(T,S)$, and $t_0 \in T$ be such that $V(\varphi, t_0) < 1$. Since we are only concerned with the valuation at $t_0$, we may assume without loss of generality that $T = \mathbb{Z}$, $t_0 = 0$, and $S$ is the successor function; in particular, that $T$ is countable. 
Let $X_0$ be the set of all real numbers $x$ such that $V(\psi, t) = x$ for some $\lanfull$-formula $\psi$ and some $t\in T$. Thus, $X_0$ is a countable subset of $[0,1]$. Let $X = (0,1) \setminus X_0$. 

We consider the bi-relational frame $\mathcal{F} = (X, T, {\leq}, S)$, where $\leq$ is the usual order on real numbers, and the bi-relational valuation $\lb\,\cdot\,\rb$ given by 
\begin{equation} \label{eqInductionRealToKripke}
(x, t) \in \lb  p \rb \text{ if and only if } V(p, t)> x.
\end{equation}
We may then prove by induction that \eqref{eqInductionRealToKripke} holds for arbitrary $\lanfull$-formulas $\psi$ and arbitrary $x \in X$ and $t \in T$. Then letting $x\in X$ be such that $V(\varphi, t_0)< x$ (this exists because $X_0$ is countable), we have $(x, t_0) \not\in \lb \varphi\rb$, so that $(\mathcal{F}, \lb\,\cdot\,\rb)$ is a bi-relational countermodel for $\varphi$, as needed.
\end{proof}

\begin{lemma}\label{LemmaKripkeToReal}
Suppose that $\varphi$ is an $\lanfull$-formula that is not bi-relationally valid. Then $\varphi$ is not real valid.
\end{lemma}
\begin{proof}
Suppose that there is a bi-relational frame $\mathcal{F} = (W, T, \allowbreak{\leq}, S)$ and a valuation $\lb\,\cdot\,\rb$ such that $(w,t_0) \not\in \lb\varphi\rb$. As above, we may assume that $T = \mathbb{Z}$, that $S$ is the successor function, and that $t_0 = 0$.
By a routine downward L\"owenheim--Skolem-type argument, we may assume that $W$ is countable.\footnote{Build a suborder $(W^*,{\leq}^*)$ of $(W,{\leq})$ by induction on the structure of $\varphi$ as follows: start with $\{w\}$ and inductively decompose $\varphi$ according to its outermost connective. When considering subformulas $\psi$ of the form $\varphi_0 \imp \varphi_1$ or $\varphi_0 \dimp \varphi_1$, we add to $W^*$ (if necessary), for each world $w^*$ in the suborder being built and for each $s \in T$, a new element of $W$  witnessing the quantifier in the definition of $\lb\psi\rb$ in a way that ensures that `$(w^*, s) \in \lb\psi\rb$' holds in $(W,T,{\leq}, S)$ if and only if it holds in $(W^*, T, \leq^*, S)$.}

We define a binary relation on $\lanfull \times T$ by 
\[(\psi, t_1) \leq (\chi, t_2) \text{ if and only if } \forall w \in W\, ((w, t_1) \in \lb \psi \rb \implies (w, t_2) \in \lb \chi \rb),\]
that is, if the valuation of $\psi$ at $t_1$ is contained in that of $\chi$ at $t_2$.
Since valuations of formulas are downward closed in their first coordinates, this reflexive transitive relation is total (any two elements are comparable). Let $(L, \leq)$ denote the linear order of all equivalence classes $[\psi, t]$ under $\leq$, and let $0$ and $1$ denote the equivalence classes of $[\bot, t]$ and $[\top, t]$, respectively (these are independent of the choice of $t$).

We claim that this linear order respects valuations, in the sense that it satisfies the following properties for each $t \in T$:
\begin{itemize}
\item $0$ is the least element and $1$ is the greatest element of $L$;
\item $[\psi\wedge\chi, t] = \min\{[\psi, t], [\chi, t]\}$
and $[\psi\vee\chi, t] = \max\{[\psi, t], [\chi, t]\}$;
\item $[\chi \imp \psi, t] = 1$ if $[\chi, t] \leq [\psi, t]$, and $[\chi \imp \psi, t] = [\psi, t]$ otherwise;
\item $[\chi \dimp \psi, t] = 0$ if $[\chi, t] \leq [\psi, t]$, and $[\chi \dimp \psi, t] = [\chi, t]$ otherwise;
\item $[\nx \psi, t] = [\psi, t+1]$ and $[\y \psi, t] = [\psi, t-1]$;
\item
 $[\nec \psi, t] = \inf_{n<\infty} [\psi, t+n]$ and $[\has \psi, t] = \inf_{n<\infty} [\psi, t-n]$;
\item $[\varphi \until \psi, t] = \sup_{n<\infty}\min\{ [\varphi, t], \dots, [\varphi, t+(n-1)], [\psi, t+n]\}$;
\item $[\varphi \since \psi, t] = \sup_{n<\infty}\min\{ [\varphi, t], \dots, [\varphi, t-(n-1)], [\psi, t-n]\}$.
\end{itemize}
These properties easily follow from the definitions.

Now, $(L, \leq)$ is a countable linear order with endpoints, so it can be continuously (with respect to the order topology) embedded into the interval $[0,1]$ in such a way that the images of $0$ and $1$ are, respectively, $0$ and $1$. 
Let $\rho$ be such an embedding. (We refer the reader to for example the proof of \cite[Theorem 5.1]{BPZ07} for an explicit construction of such a $\rho$.) We define a real valuation $V$ by setting $V(\psi, t) = \rho([\psi, t])$. By the properties above, $V$ is indeed a real valuation and $V(\varphi, t_0) = \rho([\varphi,t_0]) < 1$.
\end{proof}

\begin{theorem}\label{thm:equal}
$\gtlreal = \gtlrel$.
\end{theorem}

\section{Labelled systems and quasimodels}\label{SecNDQ}

Our decidability proof for the set of G\"odel temporal validities $\gtlreal$ is based on (nondeterministic) quasimodels, originally introduced in \cite{Fernandez09} for \emph{dynamic topological logic,} a classical predecessor of \emph{intuitionistic temporal logic,} for which quasimodels were also used in \cite{F-D18}.
As the bi-relational semantics makes evident, G\"odel temporal logic is closely related to intuitionistic temporal logic. In this section we will introduce labelled spaces, labelled systems, and finally, quasimodels. Quasimodels can be viewed as a sort of temporally nondeterministic generalisation of bi-relational models.

Of course, many decidability proofs for classical modal logics are obtained via the finite model property, so it is worthwhile to first note that this strategy cannot work for $\gtlreal$ because finite model properties do not hold, in either semantics. The finite model properties we define are of the form: \emph{falsifiable} implies \emph{falsifiable in a finite model}.
It is worth remarking that in sub-classical logics, it is indeed the notion of falsifiability that is relevant, as it is falsifiability that is (by definition) the complement to validity.
However, in view of our inclusion of co-implication, we may define $\dneg\varphi\equiv \top\dimp\varphi$, and then it is not hard to check that $\dneg\varphi$ is \emph{satisfiable} (in the sense of having a non-zero truth value) if and only if $\varphi$ is falsifiable.
Thus in view of the fact that our logic is (as we will see) \textsc{pspace}-complete, validity, satisfiability, and falsifiability are all inter-reducible.

\begin{definition}\label{def:fmp}
The \define{strong finite model property} for $\gtlrel$ is the statement that if $\varphi \in \lanfull$ is falsifiable on a bi-relational model, then it is falsifiable on a bi-relational model $\mathcal  F=(W,T,{\leq},S, \lb\,\cdot\,\rb)$ where both $W$ and $T$ are finite.

The \define{order finite model property} for $\gtlrel$ is the statement that if $\varphi \in \lanfull$ is falsifiable on a bi-relational model, then it is falsifiable on a bi-relational model $\mathcal  F=(W,T,{\leq},S, \lb\,\cdot\,\rb)$ where $W$ is finite.

The \define{temporal finite model property} for $\gtlrel$ is the statement that if $\varphi \in \lanfull$ is falsifiable on a bi-relational model, then it is falsifiable on a bi-relational model $\mathcal  F=(W,T,{\leq},S, \lb\,\cdot\,\rb)$ where $T$ is finite.
\end{definition}

\begin{proposition}\label{no_finite}
None of the finite model properties for $\gtlrel$ listed in \Cref{def:fmp} hold.
In particular, $\ps (p \imp \nx p)$ is falsifiable, yet it is valid over the class of finite bi-relational models.
\end{proposition}

\begin{proof}
To see that $\ps (p \imp \nx p)$ is falsifiable, take the flow $(\mathbb Z, S)$, where $S$ is successor. Consider the model $(\mathbb Z,\mathbb Z,{\leq},S, \lb\,\cdot\,\rb)$, where $\lb p\rb = \{(n, t) \in \mathbb Z \times \mathbb Z \mid n \leq -t\}$. (See \Cref{fig:1} (left).) Then at each $(-i, i)$, the formula $p$ holds but $\nx p$ does not. Hence at each $(0, i)$ with $i \in \mathbb Z_{\geq 0}$, the formula $p \imp \nx p$ is falsified. Thus $\ps (p \imp \nx p)$ is falsified at $(0, 0)$.

To see that $\ps (p \imp \nx p)$ can only be falsified on a model $(W,T,{\leq},S, \lb\,\cdot\,\rb)$ for which both $W$ and $T$ are infinite, suppose $\ps (p \imp \nx p)$ is falsified at $(w, t)$. Then there must be a sequence $w\geq w_0 > w_1> \dots$ such that for each $i \in \mathbb Z_{\geq 0}$, the formula $p$ holds on each $(w_i, S^i(t))$ but not on $(w_i, S^{i+1}(t))$. This clearly forces $W$ to be infinite, and by downward closure of $\lb\,\cdot\,\rb$, it forces $T$ to be infinite too. (See \Cref{fig:1} (right).)
\end{proof}

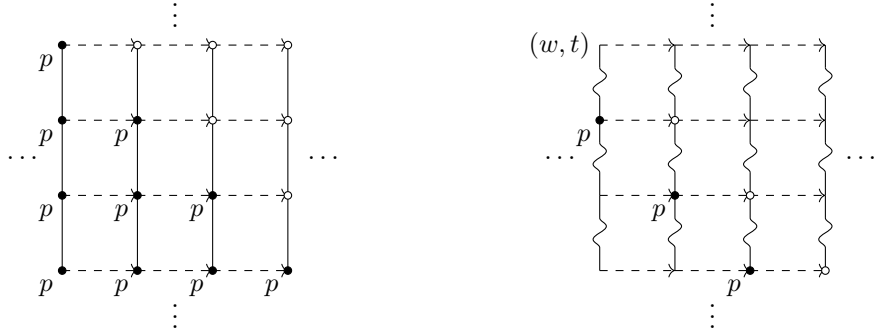
\begin{figure}\centering
\begin{tikzpicture}
\node[left] at (0,-0) {\phantom{$(w,t)$}};
\foreach \i in {3}
{
\draw (\i,0) -- (\i,-3);
\foreach \j in {\i,...,3}
{
\node[below left] at (\i,-\j) {$p$};
\draw[fill] (\i,-\j) circle (.05);
}
}
\foreach \i in {0, 1, 2}
{
\draw (\i,0) -- (\i,-3);
\foreach \j in {0,...,3}
{
\draw[dashed, ->] (\i,-\j) -- (\i+1, -\j);
}
\foreach \j in {\i,...,3}
{
\node[below left] at (\i,-\j) {$p$};
\draw[fill] (\i,-\j) circle (.05);
}
}
\foreach \i in {0,1,2}
\foreach \j in {0,...,\i}
\draw[fill=white] (\i+1,-\j) circle (.05);

\node at (3.5, -1.5) {\dots};
\node at (1.5, -3.5) {\vdots};
\node at (-.5, -1.5) {\dots};
\node at (1.5, .5) {\vdots};
\end{tikzpicture}
\hspace{2cm}
\begin{tikzpicture}
\foreach \i in {0, 1, 2}
{
\foreach \j in {0,1,2}
{
\draw (\i,-\j+0) -- (\i,-\j-.3);
\draw[decorate, decoration=snake] (\i,-\j-.3) -- (\i,-\j-.7);
\draw (\i,-\j-.7) -- (\i,-\j-1);
}
\foreach \j in {0,...,3}
{
\draw[dashed, ->] (\i,-\j) -- (\i+1, -\j);
}
}
\foreach \i in {3}
{
\foreach \j in {0,1,2}
{
\draw (\i,-\j+0) -- (\i,-\j-.3);
\draw[decorate, decoration=snake] (\i,-\j-.3) -- (\i,-\j-.7);
\draw (\i,-\j-.7) -- (\i,-\j-1);
}
}
\node at (3.5, -1.5) {\dots};
\node at (1.5, -3.5) {\vdots};
\node at (-.5, -1.5) {\dots};
\node at (1.5, .5) {\vdots};

\node[left] at (0,-0) {$(w,t)$};
\node[below left] at (0,-1) {$p$};
\draw[fill] (0,-1) circle (.05);
\draw[fill=white] (1,-1) circle (.05);
\node[below left] at (1,-2) {$p$};
\draw[fill] (1,-2) circle (.05);
\draw[fill=white] (2,-2) circle (.05);
\node[below left] at (2,-3) {$p$};
\draw[fill] (2,-3) circle (.05);
\draw[fill=white] (3,-3) circle (.05);
\end{tikzpicture}
\caption{Left: A bi-relational model falsifying $\ps (p \imp \nx p)$; right: $W$ and $T$ are necessarily infinite. Time is represented by the \emph{horizontal} axes and truth by the \emph{vertical}.}
\label{fig:1}
\end{figure}

The same example as in \Cref{no_finite} shows that under real-valued semantics it is also the case that some formulas can only be falsified on an infinite flow with infinitely many realised truth values, as falsification of $\ps (p \imp \nx p)$ forces $V(p,t ) >V(p ,S(t) )$ for all $t$.

Note that we have refuted all these finite model properties without using many of the connectives (in particular without $\dimp$ and without any past connectives), thus in fact proving the stronger result that the finite model properties fail for the fragment standardly used for (future time) linear temporal logic.

We now begin to introduce the structures we will use to mitigate the failure of these finite model properties.

\begin{definition}\label{definition:type}
Given a set $\Sigma\subseteq\lanfull$ that is closed under subformulas, we say that $\Phi\subseteq \Sigma$ is a \define{$\Sigma$-type} if the following occur.
\begin{enumerate}


\item\label{type1} If $\varphi\wedge\psi\in \Sigma$, then $\varphi\wedge\psi\in \Phi$ if and only if $\varphi \in \Phi$ and $\psi\in \Phi$.

\item\label{type2} If $\varphi\vee\psi\in \Sigma$, then $\varphi\vee\psi\in \Phi$ if and only if either $\varphi\in\Psi$ or $\psi\in \Phi$.

\item\label{type3} If $\varphi\imp \psi\in \Sigma$, then

\begin{enumerate}

\item\label{type3a}  $\varphi\imp \psi\in\Phi$ implies either $\varphi\not \in \Phi$ or $\psi\in\Phi$,

\item\label{type3b} $\psi \in \Phi$ implies  $\varphi\imp \psi \in\Phi$.

\end{enumerate}

\item\label{type4} If $\varphi\dimp \psi\in \Sigma$, then

\begin{enumerate}

\item\label{type4a} $\varphi\dimp \psi\in\Phi$ implies  $\varphi  \in \Phi$, 

\item  $\varphi\in\Phi$ and $\psi \notin \Phi$ implies  $\varphi\dimp \psi \in\Phi$.

\end{enumerate}



\end{enumerate}


The set of $\Sigma$-types will be denoted by $\type\Sigma$.
\end{definition}

A partially ordered set $(A,\leq)$ is \define{locally linear} if it is a disjoint union of linear posets.
If $a,b\in A$, we write $a\compa b$ if $a\leq b$ or $b\leq a$.
We call the set $\{b\in A:b\compa a\}$ the \define{linear component} of $a$; by assumption, the linear components partition $A$.
 

\begin{definition}\label{frame}
Let $\Sigma\subseteq\lanfull$ be closed under subformulas.
A \define{$\Sigma$-labelled space} is a triple $\mathcal  W= ( W,{\leq}_\mathcal  W ,\ell_\mathcal  W )$, where $( W ,{\leq}_\mathcal  W )$ is a locally linear poset and $\ell\colon W \to \type\Sigma$ an inversely monotone function, in the sense that \[w\leq_\mathcal W  v \text{ implies } \ell_\mathcal  W(w) \supseteq \ell_\mathcal  W(v),\] and such that for all $w\in W$:
\begin{itemize}
\item whenever $\varphi\imp \psi\in \Sigma \setminus \ell_\mathcal  W(w)$, there is $v\leq w$ such that $\varphi\in \ell_\mathcal  W(v)$ and $\psi \not \in \ell_\mathcal  W(v)$;
\item
whenever $\varphi\dimp \psi\in  \ell_\mathcal  W(w)$, there is $v\geq w$ such that $\varphi\in \ell_\mathcal  W(v)$ and $\psi\not\in \ell_\mathcal  W (v)$.
\end{itemize}

The $\Sigma$-labelled space $\mathcal  W$ \define{falsifies} $\varphi\in\mathcal L$ if $\varphi\in \Sigma\setminus \ell_\mathcal  W(w)$ for some $w\in W$.
The \define{height} of $\mathcal W$ is the supremum of all $n$ such that there is a chain $w_1 <_\mathcal  W w_2 <_\mathcal  W \dots <_\mathcal  W  w_n$.
\end{definition}

If $ \mathcal W$ is a labelled space, elements of its domain $W$ will sometimes be called \define{worlds}.\footnote{Note that the conditions on a labelled space ensure that every type $\ell_\mathcal W(w)$ used to label a world is necessarily `consistent' in the sense that it cannot contain $\bot$ (which implies further that it can never contain both $\varphi$ and $\neg \varphi$).}
When clear from the context we will omit subscripts and write, for example,~$\leq$ instead of $\leq_\mathcal  W$.

We now enrich labelled spaces with a relation capturing \emph{temporal} information.

\begin{definition}
Recall that a subset $S$ of a poset $(P,\leq)$ is \define{convex} if $s \in S$ whenever $a,b \in S$ and $a\leq s\leq b$. 
In this article, a \define{convex relation} 
 between posets $(A,{\leq}_A)$ and $(B,{\leq}_B)$ is a binary relation $R\subseteq A\times B$ such that for each $x \in A$ the image set $\{y \in B \mid x \mathrel R y\}$ is convex with respect to $\leq_B$, and for each $y \in B$ the preimage set $\{x \in A \mid x \mathrel R y\}$ is convex with respect to $\leq_A$.
 
 The relation $R$ is \define{fully confluent} if it validates the four following conditions:
\begin{description}

\item[forth--down]\label{forward--down}\phantom{.}\newline if $x \leq _A x' \mathrel R y'$ there is $y$ such that $x \mathrel R y \leq_B y'$;

\item[forth--up]\label{forward--up}\phantom{.}\newline if $x' \geq _A x \mathrel R y$ there is $y'$ such that $x' \mathrel R y' \geq_B y$; 

\item[back--down]\label{backward--down}\phantom{.}\newline if $x' \mathrel R y' \geq_B y$ there is $x$ such that $x'  \geq _A x \mathrel R y$;

\item[back--up]\phantom{.}\newline if $x \mathrel R y \leq_B y'$ there is $x'$ such that $x \leq _A x' \mathrel R y'$.

\end{description}
\end{definition}

\begin{figure}
\centering
\begin{tikzpicture}
\draw[line width=.25mm,-Straight Barb] (0,0) -- (0,2) -- (2,2);
\draw[line width=.25mm,dashed,-Straight Barb] (0,0) -- (2,0);
\draw[line width=.25mm,dashed](2,0)  -- (2,2);
\node at (0,0)[circle, fill, inner sep = .8pt]{};
\node at (0,2)[circle, fill, inner sep = .8pt]{};
\node at (2,0)[circle, fill, inner sep = .8pt]{};
\node at (2,2)[circle, fill, inner sep = .8pt]{};
\node at (0,1)[anchor=east]{$\leq$};
\node at (2,1)[anchor=west]{$\leq$};
\node at (1,0)[anchor=north, align=center,minimum size=3em]{$R$\\forth--down};
\node at (1,2)[anchor=south]{$R$};
\end{tikzpicture}
\hspace{.3cm}
\begin{tikzpicture}
\draw[line width=.25mm,-Straight Barb] (0,2) -- (0,0) -- (2,0);
\draw[line width=.25mm,dashed] (2,0) -- (2,2);
\draw[line width=.25mm,-Straight Barb, dashed] (0,2) -- (2,2);
\node at (0,0)[circle, fill, inner sep = .8pt]{};
\node at (0,2)[circle, fill, inner sep = .8pt]{};
\node at (2,0)[circle, fill, inner sep = .8pt]{};
\node at (2,2)[circle, fill, inner sep = .8pt]{};
\node at (0,1)[anchor=east]{$\leq$};
\node at (2,1)[anchor=west]{$\leq$};
\node at (1,0)[anchor=north, align=center,minimum size=3em]{$R$\\forth--up};
\node at (1,2)[anchor=south]{$R$};
\end{tikzpicture}
\hspace{.3cm}
\begin{tikzpicture}
\draw[line width=.25mm,dashed,-Straight Barb] (0,2) -- (0,0) -- (2,0);
\draw[line width=.25mm,-Straight Barb] (0,2) -- (2,2);
\draw[line width=.25mm](2,0)  -- (2,2);
\node at (0,0)[circle, fill, inner sep = .8pt]{};
\node at (0,2)[circle, fill, inner sep = .8pt]{};
\node at (2,0)[circle, fill, inner sep = .8pt]{};
\node at (2,2)[circle, fill, inner sep = .8pt]{};
\node at (0,1)[anchor=east]{$\leq$};
\node at (2,1)[anchor=west]{$\leq$};
\node at (1,0)[anchor=north, align=center,minimum size=3em]{$R$\\back--down};
\node at (1,2)[anchor=south]{$R$};
\end{tikzpicture}
\hspace{.3cm}
\begin{tikzpicture}
\draw[line width=.25mm,-Straight Barb] (0,0) -- (2,0);
\draw[line width=.25mm](2,0) -- (2,2);
\draw[line width=.25mm,dashed,-Straight Barb] (0,0) -- (0,2)  -- (2,2);
\node at (0,0)[circle, fill, inner sep = .8pt]{};
\node at (0,2)[circle, fill, inner sep = .8pt]{};
\node at (2,0)[circle, fill, inner sep = .8pt]{};
\node at (2,2)[circle, fill, inner sep = .8pt]{};
\node at (0,1)[anchor=east]{$\leq$};
\node at (2,1)[anchor=west]{$\leq$};
\node at (1,0)[anchor=north, align=center,minimum size=3em]{$R$\\back--up};
\node at (1,2)[anchor=south]{$R$};
\end{tikzpicture}
\caption{Confluence conditions}\label{figure:confluence}
\end{figure}
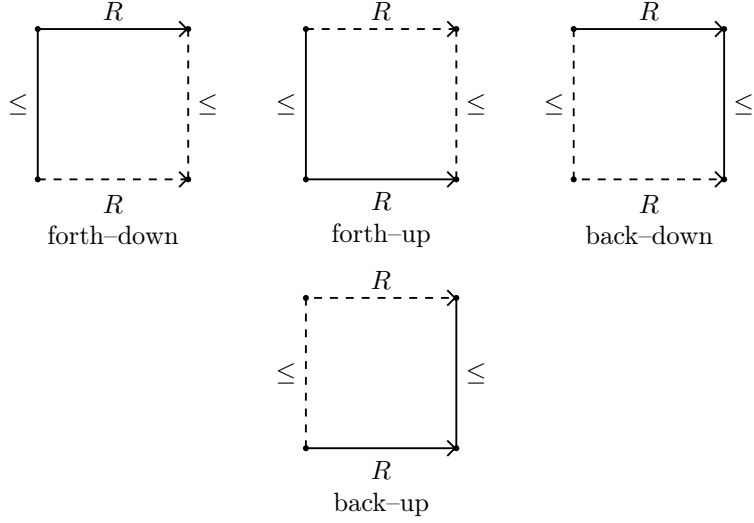
\Cref{figure:confluence} depicts the four confluence conditions that a fully confluent relation must validate.

\begin{definition}\label{compatible}
Let $\Sigma\subseteq\lanfull$ be closed under subformulas. Suppose that $\Phi,\Psi\in\type\Sigma$. The ordered pair $(\Phi,\Psi)$ is \define{sensible} if
\begin{enumerate}[label=(\arabic*)]
\item for all $\nx\varphi\in \Sigma$, we have $\nx\varphi\in \Phi$ if and only if $ \varphi\in \Psi$;
\item for all $\y\varphi\in \Sigma$, we have $\y\varphi\in \Psi$ if and only if $ \varphi\in \Phi$;
\item for all $\nec \varphi\in \Sigma$, we have $\nec \varphi\in \Phi$ if and only if $\varphi\in \Phi$ and $\nec \varphi\in \Psi$;
\item for all $\has \varphi\in \Sigma$, we have $\has \varphi\in \Psi$ if and only if $\varphi\in \Psi$ and $\has \varphi\in \Phi$;
\item for all $\varphi \until \psi \in \Sigma$, we have $\varphi \until \psi\in \Phi$ if and only if either $\psi\in\Phi$ or ($\varphi\in \Phi$ and $\varphi \until \psi \in \Psi$);
\item for all $\varphi \since \psi \in \Sigma$, we have $\varphi \since \psi\in \Psi$ if and only if either $\psi\in\Psi$ or ($\varphi\in \Psi$ and $\varphi \since \psi \in \Phi$).
\end{enumerate}
A pair $(w,v)$ of worlds in a $\Sigma$-labelled space $\mathcal  W$ is \define{sensible} if $(\ell (w),\ell (v))$ is sensible.
A \emph{relation}
$R\subseteq   W\times   W$
is \define{sensible} if every pair in $R$ is sensible. It is \define{bi-serial}, if for each $w \in   W$ there exist both $v_1 \in   W$ with $w \mathrel R  v_1$ \emph{and} $v_{-1} \in   W$ with $v_{-1} \mathrel R  w$.
Further, $R$ is \define{$\om$-sensible} if:
\begin{itemize}

\item whenever $\nec\varphi\in \Sigma\setminus\ell(w)$, there are $n\geq 0$ and $v$ such that $w \mathrel R^n v$ and $\varphi\in \Sigma\setminus\ell(v)$;
\item whenever $\has\varphi\in\Sigma\setminus\ell(w)$, there are $n\geq 0$ and $v$ such that $v \mathrel R^n w$ and $\varphi\in \Sigma\setminus \ell(v)$;
\item whenever $\varphi \until \psi \in \ell(w)$, there are $n\geq 0$ and 
$v$ such that $w \mathrel{R^n} v$ and $\psi \in \ell(v)$;\footnote{Note that because $R$ is assumed to be sensible, this condition is equivalent to the existence of $m\geq 0$ and  $v_0 \mathrel R v_1 \mathrel R\dots\mathrel R v_m$ such that $v_0 = w$ and $\varphi\in \ell(v_0), \dots, \varphi\in \ell(v_{m-1}), \psi\in \ell(v_m)$. Similarly for the following $\since$ condition.}

\item whenever $\varphi \since \psi \in \ell(w)$, there are $n\geq 0$ and 
$v$ such that $v \mathrel{R^n} w$ and $\psi \in \ell(v)$.
\end{itemize}

A \define{labelled system} is a labelled space $\mathcal  W$ equipped with a bi-serial, fully confluent, convex, sensible relation $R_\mathcal W\subseteq  W\times  W$. If moreover $R_\mathcal W$ is $\om$-sensible, we say that $\mathcal  W$ is a \define{$\Sigma$-quasimodel}.
\end{definition}

Any bi-relational model can be regarded as a $\Sigma$-quasimodel. If $\mathcal {X} = (W,T,{\leq},\allowbreak S, \val {\,\cdot\,})$ is a bi-relational model and $x\in W \times T$, we can assign a $\Sigma$-type $\ell_\mathcal  X(x)$ to $x$ given by
$\ell_\mathcal  X(x)=\cbra\psi\in \Sigma \mid x\in \val\psi \cket.$
We also set $R_\mathcal  X=\{((w, t) , (w, S(t))) \mid w \in W,\ t \in T\}$; it is obvious that $R_\mathcal  X$ is $\om$-sensible.
Henceforth we will tacitly identify $\mathcal  X$ with its associated $\Sigma$-quasimodel.

\section{From quasimodels to bi-relational models}\label{SecGen}

If $\mathcal  Q = (   Q,{\leq}_\mathcal  Q ,\ell_\mathcal  Q, R_\mathcal Q )$ is a quasimodel
, then we cannot necessarily view $\mathcal  Q$ directly as a bi-relational model for the primary reason that $R_\mathcal  Q$ is not necessarily a function. However, we can extract bi-relational models from quasimodels via an unwinding and selection construction. More precisely, given a $\Sigma \subseteq\lanfull$ that is finite and closed under subformulas, suppose $\varphi$ is falsified on the $\Sigma$-quasimodel $\mathcal Q$. In this section we show how to obtain from $\mathcal Q$ a bi-relational model $\lm_\varphi$ falsifying $\varphi$. We call the resulting bi-relational model a \emph{limit model} of $\mathcal Q$. This proves $\gtlrel$ is sound for the class of quasimodels.

 The general idea for determinising $\mathcal Q$ is to consider infinite  paths on $\mathcal  Q$ as points in the limit model. However, we will only select paths $\vec w$ with the property that, if $\varphi \until \psi$ occurs in $\vec w$, then $\psi$ must also occur at a later time, with similar conditions for $\since$, $\nec$, and $\has$. These are the \emph{realising} paths of $\mathcal  Q$.

\begin{definition}
A \define{path} in a $\Sigma$-quasimodel $\mathcal  Q = (   Q,{\leq}_\mathcal  Q ,\ell_\mathcal  Q, R_\mathcal Q )$ is any  sequence $\<w_i\>_{\alpha <i < \beta}$ with $-\infty \leq \alpha < \beta \leq \infty$ 
${\vec w}=\< w_i\>_{i\in\mathbb Z}$
is \define{realising} if for all $i\in\mathbb Z$:
\begin{itemize}
\item
 for all $\nec\psi\in \Sigma \setminus \ell(w_i)$, there exists $j\geq i$ such that $\psi\in \Sigma \setminus \ell  (w_j)$;
\item
 for all $\has\psi\in \Sigma \setminus \ell(w_i)$, there exists $j\leq i$ such that $\psi\in \Sigma \setminus \ell  (w_j)$;
\item
 for all $\varphi \until \psi\in \ell(w_i)$, there exists $j\geq i$ such that $\psi\in \ell  (w_j)$;
\item
 for all $\varphi \since \psi\in \ell(w_i)$, there exists $j\leq i$ such that $\psi\in \ell  (w_j)$.
\end{itemize}
\end{definition}

Denote the set of realising paths by $ \domlm $, and let $\< v_i\>_{i\in\mathbb Z}\leq \< w_i\>_{i\in\mathbb Z}$ if and only if $ v_i\leq w_i $ for all $i\in\mathbb Z$.
The worlds of the limit model will be a  subset of $\lm$ that is linearly ordered (with respect to $\leq$).

Given our $\Sigma \subseteq\lanfull$ that is finite and closed under subformulas and a formula $\varphi$ falsified in $\mathcal Q$, the limit model $\lm_\varphi$ will be of the form $(W, \mathbb Z, {\leq}, S, \allowbreak\val {\,\cdot\,})$, where $W \subseteq \domlm$, the flow function $S \colon \mathbb Z \to \mathbb Z$ is successor, and $\val p = \{(\vec w ,i) \mid p \in \ell(w_i)\}$ (extended to compound formulas in accordance with \Cref{DefKSem}). We now describe how to select the linearly ordered subset of realising paths $W$.

\begin{definition}
A \define{finite grid} is a finite set of paths in $\mathcal Q$ with uniform index bounds $ l < k \in \mathbb Z$ that is linearly ordered (by the pointwise ordering). 

A finite grid $P'$ of paths with bounds $l' < k'$ \define{extends} a finite grid $P$ of paths with bounds $k$, if $l' \leq l < k \leq k'$ and $P \subseteq P'|_{l,k}$, where $P'|_{l,k}$ is the restriction of the paths in $P'$ to their segments within the bounds $l$ and $k$.
\end{definition}

\begin{definition}[defects]
 Let $P$ be a finite grid.
\begin{itemize}
\item
A \define{$\nec$-defect} of $P$ is a pair $((w_i)_{l<i<k}, \nec\varphi) \in P \times \Sigma$ such that $\nec\varphi \not\in \ell(w_{k-1})$, but $\varphi \in \ell(w_{k-1})$.
\item
A \define{$\has$-defect} of $P$ is a pair $((w_i)_{l<i<k}, \has\varphi) \in P \times \Sigma$ such that $\has\varphi \not\in \ell(w_{l+1})$, but $\varphi \in \ell(w_{l+1})$.
\item
An \define{$\until$-defect} of $P$ is a pair $((w_i)_{l<i<k}, \varphi \until \psi) \in P \times \Sigma$ such that $\varphi \until \psi \in \ell(w_{k-1})$, but $\psi \not\in \ell(w_{k-1})$.
\item
A \define{$\since$-defect} of $P$ is a pair $((w_i)_{l<i<k}, \varphi \since \psi) \in P \times \Sigma$ such that $\varphi \since \psi \in \ell(w_{l+1})$, but $\psi \not\in \ell(w_{l+1})$.
\item
An \define{$\imp$-defect} of $P$ is a triple $((w_i)_{l<i<k}, \varphi \imp\psi, j ) \in P \times \Sigma \times \{l+1,\dots,k-1\}$ such that $\varphi \imp\psi \not\in \ell(w_{j})$, but there is no $(v_i)_{l<i<k} \leq (w_i)_{l<i<k}$ also in $P$ such that $\varphi \in \ell(v_{j})$, but $\psi \not\in \ell(v_{j})$.
\item
A \define{$\dimp$-defect} of $P$ is a triple $((w_i)_{l<i<k}, \varphi \dimp\psi, j ) \in P \times \Sigma \times \{l+1,\dots,k-1\}$ such that $\varphi \dimp\psi \in \ell(w_{j})$, but there is no $(v_i)_{l<i<k} \geq (w_i)_{l<i<k}$ also in $P$ such that $\varphi \in \ell(v_{j})$, but $\psi \not\in \ell(v_{j})$.
\item
A \define{seriality defect} is a pair $((w_i)_{l<i<k}, b) \in P \times \{-,+\}$.
\end{itemize}
\end{definition}

Note that because $\Sigma$ is finite, any finite grid has a finite number of defects. We select the set $W \subseteq \domlm$ as follows. We maintain a (finite) first-in-first-out queue $D$ of defects and a finite grid $P$ that is extended each time we process a defect. (We will ensure that all constituents of $D$ continue to be defects of the grid after each update.) Choose some $w \in Q$ falsifying $\varphi$. We initialise the grid to the single sequence $(w)$ (of length 1, using the index $0$) and add all defects of this grid to $D$.

Using the properties of quasimodels, we may extend $P$ to a new grid $P'$, so that the defect at the head of the queue $D$ is not a defect of $P'$. (Because of seriality defects, $D$ is never empty.)
The construction of $P'$ is realised via a case-by-case analysis. \begin{itemize}
\item
If the defect is an $\until$-defect $((w_i)_{l<i<k}, \varphi \until \psi)$ of $P$, then because $R=R_\mathcal  Q$ is $\om$-sensible  we know that there exist $j> 0$ and $v\in R^{j}(w_{k-1})$ such that $\varphi\in \ell(v)$.  We then define $k'=k+j$ and choose
\[w_{k -1}\mathrel R w_{k}\mathrel R \hdots \mathrel R  w_{k'-1}=v.\]
By the forth--up confluence property, we can extend every $ (v_i)_{l<i<k} > (w_i)_{l<i<k}$ in $P$ to a sequence with upper index bound $k'$ in a way that preserves the (linear) ordering on sequences. Similarly for every $ (v_i)_{l<i<k} < (w_i)_{l<i<k}$ in $P$ using the forth--down confluence property. 
\item
If the defect is a $\nec$-defect $((w_i)_{i<k}, \nec\varphi)$, we can find $j> 0$ and $v\in R^{j}(w_{k-1})$ such that $\varphi\not\in \ell(v)$, and proceed as for the $\until$ case.

\item
The cases of $\since$- and $\has$-defects are the temporal duals of $\until$- and $\nec$-defects, respectively.

\item
If the defect is an $\imp$-defect $((w_i)_{l<i<k}, \varphi \imp\psi, j )$, choose $v_j < w_j$ such that $\varphi \in \ell(v_j)$ and $\psi \not \in \ell(v_j)$. 
Let $(u_i)_{l<i<k}$ be the minimum sequence in $P$ with $u_j > v_j$ and $(t_i)_{l<i<k}$ be the maximum sequence in $P$ with $t_j < v_j$, if it exists.
We will assume that $(t_i)_{i<k}$ is defined, since the case where $v_j$ is not bounded below is similar but simpler.
We will complete the sequence $(v_i)_{l<i<k}$ so that $(t_i)_{l<i<k} < (v_i)_{l<i<k} < (u_i)_{l<i<k}$. Suppose we have defined $v_{j'}$ for $j\leq j' < k-1$. To define $v_{j'+1}$, choose $y$ with $v_{j'} \mathrel R y$ and $y \leq u_{j'+1}$, which exists by forth--up confluence. If $y \geq t_{j'+1}$ we can set $v_{j'+1} = y$ and we are done. Otherwise, by local linearity of $\leq$, we have $y < t_{j'+1}$. In this case, choose $z$ with $v_{j'} \mathrel R z$ and $z \geq t_{j'+1}$, which exists by forth--down confluence. Then $y < t_{j'+1} \leq z$, so as the image set of $v_{j'}$ under $R$ is convex (by convexity of $R$), we have $v_{j'} \mathrel R t_{j'+1}$ and we can set $v_{j'+1} = t_{j'+1}$. In this way we can define $(v_i)_{l<i<k}$ inductively for all indices greater than $j$. The process for indices less than $j$ is similar, using back--up and back--down confluence and the convexity of preimage sets under $R$. By construction, $(v_i)_{l<i<k}$ sits strictly between  $(t_i)_{l<i<k}$ and $(u_i)_{l<i<k}$.

\item 

The case of a $\dimp$-defect is similar to that of an $\imp$-defect, except that we choose $v_j>w_j$.

\item
If the defect is a seriality defect $((w_i)_{l<i<k},+)$, we extend $(w_i)_{l<i<k}$ to $(w_i)_{l<i<k+1}$, relying on the fact that $R$ is serial, and then we extend the other sequences to length $k+1$ using  forth--up and forth--down confluence. A seriality defect $((w_i)_{l<i<k},-)$ is analogous: we extend $(w_i)_{l<i<k}$ to $(w_i)_{l-1<i<k}$, relying on converse seriality of $R$.
\end{itemize}
Next we must update $D$ so that it contains all defects of $P'$, and all the elements of $D$ that are not defects of $P'$ (which by design includes the head of $D$) are removed.
Finally, we update $P:=P'$.

We set $W$ to be the limit of this sequence of finite grids. More precisely, let the sequence of grids be $P_0, P_1,\dots$, containing paths with bounds $(l_0,k_0), (l_1,k_1)\dots$, respectively. Then $W$ is the set of paths $(w_i)_{i\in\mathbb Z}$ for which there exists $N$ such that for all $n>N$ the finite segment $(w_i)_{l_n<i<k_n}$ is in $P_n$. (This gives us the limit we would expect and forbids `diagonal' sequences whose finite segments appear in $P_0, P_1,\dots$ merely infinitely often.)

By construction, $(W, \leq)$ is a linearly ordered set, and for each $p \in \mathbb P$, the set $\lb p \rb$ is downward closed in its first coordinate.  Thus the limit model $\lm_\varphi = (W, \mathbb Z, {\leq}, S, \val {\,\cdot\,})$ indeed defines a bi-relational model. Of course $\lm_\varphi$ is only useful if $\lb\,\cdot\,\rb$ `matches' $\ell$ on \emph{all} formulas in $\Sigma$, not just propositional variables. Fortunately, this turns out to be the case.

\begin{lemma}\label{sound}
Let $\Sigma \subseteq\lanfull$ be finite and closed under subformulas, $\mathcal  Q$ be a $\Sigma$-quasimodel, $\varphi, \psi\in \Sigma$, and $\lm_\varphi = (W, \mathbb Z, {\leq}, S,\allowbreak \val {\,\cdot\,})$ be as described above.
Then
$\lb\psi\rb=\{(\vec w, i) \mid \psi\in \ell(w_i)\}.$
\end{lemma}
\begin{proof} The proof goes by standard induction of formulas. The induction steps for $\wedge$ and $\vee$ are immediate. The cases of the remaining connectives follow straightforwardly from the construction of $W$, because every defect is eventually eliminated. Indeed every $w_i$ in the limit is also in some $P_n$, at which point all the defects associated to $w_i$ are present in the finite first-in-first-out queue $D$; thus they will eventually reach the front of $D$ and be eliminated.
\end{proof}

We obtain the main result of this section, which in particular implies that $\gtlrel$ is sound for the class of quasimodels.

\begin{proposition}\label{second}
Let $\Sigma \subseteq\lanfull$ be finite and closed under subformulas and $\mathcal Q$ be any $\Sigma$-quasimodel, and suppose $\varphi$ is falsified on $\mathcal Q$. Then there exists  a bi-relational model $\lm_\varphi$ that falsifies $\varphi$.
\end{proposition}

It is interesting to note that although we assumed in this section that $\Sigma$ was finite, this restriction can be removed. Since $\lanfull$ is countable, for an arbitrary subformula-closed $\Sigma \subseteq \lanfull$, there can only be a countable number of defects in any finite grid. Thus with appropriate scheduling all defects can be eliminated in the limit.


\section{From bi-relational models to finite quasimodels}\label{Sec:quotient}

As we noted earlier, every bi-relational model can be naturally viewed as a quasimodel. However, we wish to show that, given a finite and subformula-closed $\Sigma$, we can from each bi-relational model $\mathcal X$  produce a \emph{finite} quasimodel falsifying exactly the same formulas from $\Sigma$ as $\mathcal X$. In this section we do just this by describing a quotient construction that given an arbitrary $\Sigma$-labelled system produces a new $\Sigma$-labelled system $\mathcal Q$. This construction does not require $\Sigma$ to be finite, but when it is, $\mathcal Q$ is finite. Since bi-relational models can be viewed as $\Sigma$-labelled systems, our quotient construction applies to them. We then observe that the construction preserves $\om$-sensibility, so that when applied to a bi-relational model, we do indeed obtain a quasimodel. An illustrative example will follow (\Cref{example:quotient}).

 The process of constructing of $\mathcal Q$ consists of two steps.	The first is to take a quotient to obtain a $\Sigma$-labelled space equipped with a bi-serial, fully confluent, sensible relation. The second step is to enlarge the equipped relation to make it also convex, so all the conditions for a $\Sigma$-labelled space hold.

Let $\Sigma$ be a subformula-closed subset of $\lanfull$, and let $\mathcal {W} = (W,{\leq},\ell, R)$ be a $\Sigma$-labelled system. 
For $w\in W$, define $L_\mathcal W(w) = \cbra \ell(v) \mid v \compa w \cket$. 
Define the equivalence relation $\sim$ on $W$ by \[w \sim v \iff (\ell(w), L(w)) = (\ell(v), L(v)).\]
If $\Sigma$ is finite, then clearly $W / {\sim}$ is finite.\footnote{Note that if $\mathcal W$ is a bi-relational model, then $\sim$ is the largest relation that is simultaneously a bisimulation with respect to the relations $\leq$ and $\geq$, with $\Sigma$ treated as the set of atomic propositions that bisimilar worlds must agree on.}

Now define a partial order $\leq_\mathcal Q$ on the equivalence classes $W /{\sim}$ of $\sim$ by
\[[w] \leq_\mathcal Q [v] \iff L(w) = L(v)\text{ and }\ell(w) \supseteq \ell(v),\]
noting that this is well-defined and is indeed a partial order.

Since each set $L(w)$ can be linearly ordered by inclusion and $\ell(w) \in L(w)$, the poset $(W / {\sim}, \leq_\mathcal Q)$ is a disjoint union of linear posets. By defining $\ell_\mathcal Q$ by 
$\ell_\mathcal Q([w]) = \ell(w)$
we obtain a $\Sigma$-labelled space $(W/ {\sim}, \leq_\mathcal Q, \ell_\mathcal Q)$; it is not hard to check that this labelling is inversely monotone and that the clauses for $\imp$ and $\dimp$ in the definition of a labelled space hold with this labelling.

Now define the binary relation $R_\mathcal Q$ on $W / {\sim}$ to be the smallest relation such that $w\mathrel R v $ implies $[w] \mathrel R_\mathcal Q [v]$. 

\begin{lemma}
The relation $R_\mathcal Q$ is bi-serial, sensible, and fully confluent.
\end{lemma}

\begin{proof}
It is clear that $R_\mathcal Q$ is bi-serial and sensible.

For confluence, suppose $[w] \mathrel R_\mathcal Q [v]$, where we may assume that $w \mathrel R v$. To see that the forth--up condition holds, suppose further that $[w] \leq_\mathcal Q [u]$. Then as $l(u) \in L(u) = L(w)$ there is some $u' \geq w$ with $[u] = [u']$. By forth--up confluence of $R$, there exists $z$ with $u' \mathrel R z$ and $v \leq z$. Then we have $[u'] \mathrel R_\mathcal Q [z]$ and $[v] \leq_\mathcal Q [z]$, as required for the forth--up condition. The proofs of the remaining three confluence conditions are entirely analogous.
\end{proof}

As promised, we now have a $\Sigma$-labelled space equipped with a bi-serial, fully confluent, sensible relation. We now transform this structure into a $\Sigma$-labelled system by making the equipped relation convex by fiat.

Define $R^+_\mathcal Q$ by $X \mathrel R^+_\mathcal Q Y$ if and only if there exist $X_1 \leq_\mathcal Q X \leq_\mathcal Q X_2$ and $Y_1 \leq_\mathcal Q Y \leq_\mathcal Q Y_2$ such that $X_2 \mathrel R_\mathcal Q Y_1$ and $X_1 \mathrel R_\mathcal Q Y_2$. Now define $\mathcal Q = (W / {\sim}, \leq_\mathcal Q, \ell_\mathcal Q, R^+_\mathcal Q)$.

\begin{proposition}\label{prop:quotient_system}
The structure $\mathcal Q$ is a $\Sigma$-labelled system.
\end{proposition}

\begin{proof}
We already know that $(W/ {\sim}, \leq_\mathcal Q, \ell_\mathcal Q)$ is a $\Sigma$-labelled space. First we must check that the enlarged relation $R^+_\mathcal Q$ is still bi-serial, sensible, and fully confluent. Bi-seriality is clear, since this is a monotone condition.

To see that $R^+_\mathcal Q$ is sensible, suppose $ X \mathrel R^+_\mathcal Q Y$ and that $\nx \varphi \in \Sigma$. Take $X_1 \leq_\mathcal Q X \leq_\mathcal Q X_2$ and $Y_1 \leq_\mathcal Q Y \leq_\mathcal Q Y_2$ such that $X \mathrel R_\mathcal Q Y_1$. Then
\begin{align*}
\nx \varphi \in \ell_\mathcal Q(X) &\implies \nx \varphi \in \ell_\mathcal Q(X_1)
&&\implies \phantom{\nx}\varphi \in \ell_\mathcal Q(Y_2)
 &&\implies \phantom{\nx}\varphi \in \ell_\mathcal Q(Y)\\
 &\implies \phantom{\nx}\varphi \in \ell_\mathcal Q(Y_1)
 &&\implies \nx \varphi \in \ell_\mathcal Q(X_2)  &&\implies \nx \varphi \in \ell_\mathcal Q(X),
\end{align*}
so $ \nx \varphi \in \ell_\mathcal Q(X) \iff  \varphi \in \ell_\mathcal Q(Y)$. The cases of the remaining modalities are similar.

For the forth--down condition, suppose $X \leq_\mathcal Q X' \mathrel R^+_\mathcal Q Y'$. Then by the definition of $R^+_\mathcal Q$, there exist $X_2 \geq_\mathcal Q X'$ and $Y_1 \leq_\mathcal Q Y'$ such that $X_2 \mathrel R_\mathcal Q Y_1$. 
 Since $X \leq_\mathcal Q X' \leq_\mathcal Q X_2$, by the forth--down condition for $R_\mathcal Q$ there is some $Y \leq_\mathcal Q Y_1$ with $X \mathrel R_\mathcal Q Y$ and therefore $X \mathrel R^+_\mathcal Q Y$. Since $Y \leq_\mathcal Q Y_1 \leq_\mathcal Q Y'$, we are done. 
 The proof that the forth--up condition holds is just the order dual of that for forth--down. The proofs of the back--down and back--up conditions are similar. Thus $R^+_\mathcal Q$ is fully confluent.



Finally, we show that $R^+_\mathcal Q$ is convex. Firstly, for the image condition, if $X \mathrel R^+_\mathcal Q Y_1$ and $X \mathrel R^+_\mathcal Q Y_2$ with $Y_1 \leq_\mathcal Q Y \leq Y_2$, then by the definition of $R^+_\mathcal Q$ we can find  $X_2 \geq_\mathcal Q X$ and $Y'_1 \leq_\mathcal Q Y_1$ with $X_2 \mathrel R_\mathcal Q Y'_1$, and similarly $X_1 \leq_\mathcal Q X$ and $Y'_2 \geq_\mathcal Q Y_2$ with $X_1 \mathrel R_\mathcal Q Y'_2$. Since then $X_1 \leq_\mathcal Q X \leq_\mathcal Q X_2$ and $Y'_1 \leq_\mathcal Q Y \leq_\mathcal Q Y'_2$, by the definition of $R^+_\mathcal Q$ we conclude that $X \mathrel R^+_\mathcal Q Y$. The preimage condition is completely analogous. This completes the proof that $\mathcal Q$ is a $\Sigma$-labelled system.
\end{proof}

\Cref{prop:quotient_system} will be of independent interest later for proving the completeness of our deductive calculus. For now, however, we are only interested in the case that we had an $\om$-sensible $R$ to begin with, in which case we are certain to obtain a quasimodel, as the next lemma shows.

\begin{lemma}
If $R$ is $\om$-sensible, then $R^+_\mathcal Q$ is $\om$-sensible.
\end{lemma}

\begin{proof}
We first show that $R_\mathcal Q$ is $\om$-sensible. Suppose that $\varphi \until \psi \in \ell_\mathcal Q([w])$. Then by definition, $\varphi \until \psi \in \ell(w)$. Thus as $R$ is $\om$-sensible, there exists $v \in R^n(w)$ such that $\psi \in \ell(v)$ for some $n > 0$. It follows that $[v] \in R_\mathcal Q^n([w])$ and $ \psi \in \ell_\mathcal Q([v])$. Similar reasoning applies when we suppose that $\varphi \since \psi \in \ell_\mathcal Q([w])$, $\nec \varphi \not\in \ell_\mathcal Q([w])$, or $\has \varphi \not\in \ell_\mathcal Q([w])$; thus $R_\mathcal Q$ is $\om$-sensible. It now follows immediately that $R^+_\mathcal Q$ is $\om$-sensible, since the conditions for a relation to be $\om$-sensible are monotone.
\end{proof}

\begin{example}\label{example:quotient}
We return to the example seen in the proof of \Cref{no_finite}. Let $\Sigma$ be the set of all subformulas of $\ps (p \imp \nx p)$. So $\Sigma = \{\ps (p \imp \nx p),p \imp \nx p,\nx p, p\}$. We take as bi-relational model $\mathcal X$ the model $(\mathbb Z,\mathbb Z,{\leq},S, \lb\,\cdot\,\rb)$ seen in the proof of \Cref{no_finite} and depicted again on the left of \Cref{fig:quotient}. Examining $\mathcal X$, we see that there are three possible values that $\ell_\mathcal X (x)$ can take: $\Sigma$, $\{p\}$, or $\varnothing$. These values are depicted in the centre of \Cref{fig:quotient}. There is only one possible value for $L_\mathcal X(x)$: $\{\Sigma, \{p\}, \varnothing\}$, so the quotient quasimodel $\mathcal Q$ is linearly ordered by $\leq_\mathcal Q$. The quasimodel $\mathcal Q$ is depicted on the right of \Cref{fig:quotient} (with the constant $L$-data omitted).
\end{example}

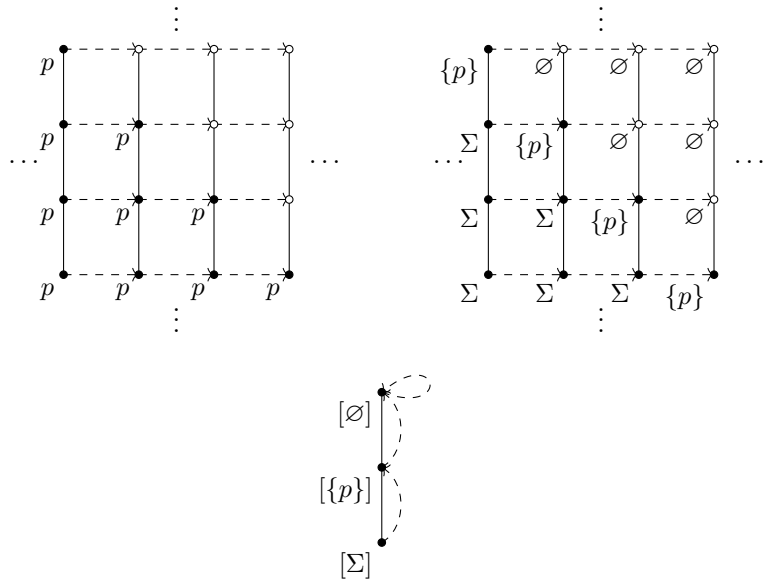
\begin{figure}\centering
\begin{tikzpicture}
\node[left] at (0,-0) {\phantom{$(w,t)$}};
\foreach \i in {3}
{
\draw (\i,0) -- (\i,-3);
\foreach \j in {\i,...,3}
{
\node[below left] at (\i,-\j) {$p$};
\draw[fill] (\i,-\j) circle (.05);
}
}
\foreach \i in {0, 1, 2}
{
\draw (\i,0) -- (\i,-3);
\foreach \j in {0,...,3}
{
\draw[dashed, ->] (\i,-\j) -- (\i+1, -\j);
}
\foreach \j in {\i,...,3}
{
\node[below left] at (\i,-\j) {$p$};
\draw[fill] (\i,-\j) circle (.05);
}
}
\foreach \i in {0,1,2}
\foreach \j in {0,...,\i}
\draw[fill=white] (\i+1,-\j) circle (.05);

\node at (3.5, -1.5) {\dots};
\node at (1.5, -3.5) {\vdots};
\node at (-.5, -1.5) {\dots};
\node at (1.5, .5) {\vdots};
\end{tikzpicture}
\hspace{.5cm}
\begin{tikzpicture}
\node[left] at (0,-0) {\phantom{$(w,t)$}};
\foreach \i in {3}
{
\draw (\i,0) -- (\i,-3);
\foreach \j in {\i,...,3}
{
\node[below left] at (\i-3,-\j+3) {$\{p\}$};
\node[below left] at (\i-2,-\j+2) {$\{p\}$};
\node[below left] at (\i-1,-\j+1) {$\{p\}$};
\node[below left] at (\i,-\j) {$\{p\}$};
\draw[fill] (\i,-\j) circle (.05);
}
}
\foreach \i in {0, 1, 2}
{
\draw (\i,0) -- (\i,-3);
\foreach \j in {0,...,3}
{
\draw[dashed, ->] (\i,-\j) -- (\i+1, -\j);
}
\foreach \j in {\i,...,3}
{
\draw[fill] (\i,-\j) circle (.05);
}
\foreach \j in {\i,...,2}
{
\node[below left] at (\i,-\j-1) {$\Sigma$};
}
}
\foreach \i in {0,1,2}
\foreach \j in {0,...,\i}
{
\node[below left] at (\i+1,-\j) {$\varnothing$};
\draw[fill=white] (\i+1,-\j) circle (.05);
}

\node at (3.5, -1.5) {\dots};
\node at (1.5, -3.5) {\vdots};
\node at (-.5, -1.5) {\dots};
\node at (1.5, .5) {\vdots};
\end{tikzpicture}
\hspace{1cm}
\begin{tikzpicture}
\draw (1.5,-.5) -- (1.5,-2.5);
\foreach \j in {0,...,2}
{
\draw[fill] (1.5,-\j-.5) circle (.05);
}
\node[below left] at (1.5,-.5) {$[\varnothing]$};
\node[below left] at (1.5,-1.5) {$[\{p\}]$};
\node[below left] at (1.5,-2.5) {$[\Sigma]$};

\draw[dashed, ->] (1.5,-0.5)  to[out=-20,in=40,distance=1cm] (1.5,-0.5);
\draw[dashed, ->] (1.5,-1.5)  to[bend right=60] (1.5,-.5);
\draw[dashed, ->] (1.5,-2.5)  to[bend right=60] (1.5,-1.5);

\node at (1.5, -3.5) {\phantom{\vdots}};

\end{tikzpicture}
\caption{Left: The bi-relational model $\mathcal X$; center: values of $\ell_\mathcal X$ when $\Sigma = \{\ps (p \imp \nx p),p \imp \nx p,\nx p, p\}$; right: the quotient quasimodel model $\mathcal Q$ obtained from $\mathcal X$ and $\Sigma$.}
\label{fig:quotient}
\end{figure}

\begin{proposition}\label{falsifies} 
Let $\mathcal {X} = (W,T,{\leq},\allowbreak S, \val {\,\cdot\,})$ be a bi-relational model and $\mathcal {Q} = ((W\times T)/{\sim},{\leq_\mathcal Q},\ell_\mathcal Q, S_{\mathcal Q}^+)$ the quotient $\Sigma$-quasimodel obtained from $\mathcal X$. Let $\varphi \in \Sigma$. Then $\mathcal X$ falsifies $\varphi$ if and only if $\mathcal Q$ falsifies $\varphi$.
\end{proposition}

\begin{proof}
We have: $\mathcal X$ falsifies $\varphi$ if and only if $\exists (w, t) \in (W \times T) \setminus \lb\varphi\rb$ if and only if $\exists [(w, t)] \in (W \times T) / {\sim}$ with $\varphi \in \Sigma \setminus \ell_\mathcal Q([(w, t)])$ if and only if $\mathcal Q$ falsifies $\varphi$.
\end{proof}

In order to use \Cref{falsifies} to prove decidability, we need to compute a bound on the size of the quasimodel $\mathcal Q$ in terms of the size of $\Sigma$, when $\Sigma$ is finite.

\begin{lemma}\label{quasi:bound}
Suppose $\Sigma$ is finite, and write $\lgt\Sigma$ for its cardinality. Then the height of $\mathcal Q$ is bounded by $\lgt\Sigma +1$, and the cardinality of the domain $(W \times T) / {\sim}$ of $\mathcal Q$ is bounded by $(\lgt\Sigma +1)\cdot 2^{\lgt \Sigma (\lgt \Sigma +1)+1}$.
\end{lemma}

\begin{proof}
Each element of the domain of $\mathcal Q$ is a pair $(\ell, L)$ where $L$ is a (nonempty) subset of $\powerset(\Sigma)$ and $\ell \in L$. Since $L$ is linearly ordered by inclusion, it has height at most $\lgt \Sigma + 1$. There are $(2^{\lgt \Sigma})^i$ subsets of $\powerset(\Sigma)$ of size $i$, so there are at most $\sum_{i = 1}^{\lgt \Sigma +1}(2^{\lgt \Sigma})^i$ distinct $L$. The sum is bounded by $2^{\lgt \Sigma (\lgt \Sigma +1)+1}$. The factor of $\lgt\Sigma +1$ corresponds to choice of an $\ell \in L$, for each $L$.
\end{proof}

Thus we have an exponential bound on the size of $\mathcal Q$. Hence any falsifiable formula is falsifiable on an effectively bounded quasimodel, and it follows that $\gtlrel$, which we know coincides with $\gtlreal$, is decidable.

\begin{theorem}\label{theorem:decide}
The logic $\gtlreal$ of $\lanfull$-formulas that are valid on all flows and the logic $\gtlrel$ of $\lanfull$-formulas that are valid on all bi-relational frames are equal and decidable.
\end{theorem}

\begin{proof}
By \Cref{thm:equal}, $\gtlreal = \gtlrel$. Since falsifiability is the complement of validity, it suffices to show that it is decidable whether a formula $\varphi$ is falsifiable over the class of all bi-relational frames. Let $\Sigma$ be the set of subformulas of $\varphi$. If $\varphi$ is falsifiable in a $\Sigma$-quasimodel of size at most $(\lgt\Sigma +1)\cdot 2^{\lgt \Sigma (\lgt \Sigma +1)+1}$, then by \Cref{second}, $\varphi$ is falsified in a bi-relational frame. Conversely, if $\varphi$ is falsified in a bi-relational frame, then by \Cref{falsifies} and \Cref{quasi:bound}, $\varphi$ is falsified in a $\Sigma$-quasimodel of size at most $(\lgt\Sigma +1)\cdot 2^{\lgt \Sigma (\lgt \Sigma +1)+1}$.
Hence it suffices to check falsifiability of $\varphi$ on the set of all $\Sigma$-quasimodels of size at most $(\lgt\Sigma +1)\cdot 2^{\lgt \Sigma (\lgt \Sigma +1)+1}$. It is clear that this check can be carried out within a computable time bound; hence the problem is decidable.
\end{proof}

\section{A Hilbert-style deductive calculus}\label{secAx}

Now begins the second half of this article, in which we prove that a certain Hilbert-style deductive calculus is sound and complete for the validities of G\"odel temporal logic (under both real-valued and bi-relational semantics). We start by presenting the calculus. It is obtained by adapting the standard axioms and inference rules of $\ltl$ \cite{temporal}, as well as their order dual and temporal dual versions.

\begin{definition}\label{defLogbasic}
The logic $\gtl$ is the least set of $\mathcal L$-formulas closed under the following axioms and rules.

\begin{enumerate}[wide, labelwidth=!, labelindent=0pt,label=\textsc{ \Roman*},ref=\textsc{\roman*}]
\item\label{ax00Intu} \textbf{All substitution instances of intuitionistic tautologies (see e.g.~\cite{MintsInt})}
\item\label{ax01Taut}\textbf{Axioms and rules of Heyting--Brouwer logic:}
\begin{multicols}3
\begin{enumerate}[label=\textsc{\alph*},ref={\textsc{\roman{enumi}}.\textsc{\alph*}}]
\item \label{axco01} $\varphi \imp (\psi \vee (\varphi \dimp \psi))$


\item\label{axDimpMon} $\dfrac{\varphi \imp \psi }{(\varphi \dimp \theta) \imp (\psi \dimp \theta)}$

\item\label{axDimpDis} $\dfrac{\varphi \imp \psi \vee \gamma}{(\varphi \dimp \psi) \imp \gamma}$

\end{enumerate}
\end{multicols}

\item \textbf{Linearity axioms:}
\begin{multicols}2
\begin{enumerate}[label=\textsc{ \alph*},ref={\textsc{\roman{enumi}}.\textsc{\alph*}}]
\item \label{axGodel} $(\varphi \imp \psi) \vee (\psi \imp \varphi)$ 
\item \label{axcoGodel} $\neg ((\varphi \dimp \psi) \wedge (\psi \dimp \varphi))$ 
\end{enumerate}
\end{multicols}

\item \textbf{Temporal axioms: }
\begin{multicols}{2}
\begin{enumerate}[label=\textsc{ \alph*},ref={\textsc{\roman{enumi}}.\textsc{\alph*}}]
\item\label{ax02Bot} $\neg \nx \bot$
\item\label{ax04NexVee} $\nx ( \varphi \vee \psi ) \imp (\nx \varphi \vee\nx \psi)$
\item\label{ax03NexWedge} $ (\nx \varphi \wedge\nx \psi) \imp \nx ( \varphi \wedge \psi ) $
\item\label{ax05KNext} $\nx( \varphi \imp \psi ) \iiff (\nx\varphi \imp \nx\psi)$
\item\label{ax06KBox} $\nec ( \varphi \imp \psi ) \imp (\nec \varphi \imp \nec \psi)$
\item\label{ax07:K:Dual} $\nec ( \varphi \imp \psi ) \imp (\theta \until \varphi \imp \theta \until \psi)$
\item\label{ax08:K:Dual} $\nec ( \varphi \imp \psi ) \imp (\varphi \until \theta \imp \psi \until \theta)$
\item\label{ax09BoxFix} $\nec \varphi \imp  \varphi\wedge \nx \nec \varphi$

\item\label{ax10DiamFix} $\psi\vee (\varphi \wedge \nx(\varphi \until \psi)) \imp \varphi \until \psi$
\item\label{ax11:ind:1} $\nec ({ \varphi \imp \nx \varphi } )\imp ({ \varphi \imp \nec \varphi })$
\item\label{ax12:ind:2}
$\nec ({ \psi\wedge \nx \varphi \imp \varphi})\imp ({ \psi\until \varphi \imp \varphi } )$
\end{enumerate}
\begin{enumerate}[label=\textsc{ \alph*}$'$,ref={\textsc{\roman{enumi}}.\textsc{\alph*}$'$}]
\item\label{ax02Bot'} $\neg \y \bot$
\item\label{ax04NexVee'} $\y ( \varphi \vee \psi ) \imp (\y \varphi \vee\y \psi)$
\item\label{ax03NexWedge'} $ (\y \varphi \wedge\y \psi) \imp \y ( \varphi \wedge \psi ) $
\item\label{ax05KNext'} $\y( \varphi \imp \psi ) \iiff (\y\varphi \imp \y\psi)$
\item\label{ax06KBox'} $\has ( \varphi \imp \psi ) \imp (\has \varphi \imp \has \psi)$
\item\label{ax07:K:Dual'} $\has ( \varphi \imp \psi ) \imp (\theta\since \varphi \imp \theta\since \psi)$
\item\label{ax08:K:Dual'} $\has ( \varphi \imp \psi ) \imp (\varphi\since \theta \imp \psi \since \theta)$
\item\label{ax09BoxFix'} $\has \varphi \imp  \varphi\wedge \y \has \varphi$

\item\label{ax10DiamFix'} $\psi\vee (\varphi \wedge \y(\varphi \since \psi)) \imp \varphi \since \psi$
\item\label{ax11:ind:1'} $\has ({ \varphi \imp \y \varphi } )\imp ({ \varphi \imp \has \varphi })$
\item\label{ax12:ind:2'} $\has ({ \psi\wedge \y \varphi \imp \varphi})\imp ({ \psi\since \varphi \imp \varphi } )$
\end{enumerate}
\end{multicols}

\item\label{axConnection}{\textbf{Connection axioms}}
\begin{multicols}2
\begin{enumerate}[label=\textsc{\alph*},ref={\textsc{\roman{enumi}}.\textsc{\alph*}}]
\item\label{temporal1}$\varphi \iiff \nx\y\varphi $
\end{enumerate}
\begin{enumerate}[label=\textsc{ \alph*}$'$,ref={\textsc{\roman{enumi}}.\textsc{\alph*}$'$}]
\item\label{temporal2}$\varphi \iiff \y\nx\varphi $
\end{enumerate}
\end{multicols}



\item\label{group4} \textbf{Standard modal rules:}
\begin{multicols}5
\begin{enumerate}[label=\textsc{mp}]
\item\label{ax13MP}  
$\dfrac{\varphi,\ \varphi\imp \psi}\psi$
\end{enumerate}
\begin{enumerate}[label=$\textsc{nec}_\nx$]
\item\label{ax14NecCirc} $\dfrac\varphi {\nx\varphi}$
\end{enumerate}
\begin{enumerate}[$\textsc{nec}_\y$]
\item\label{ax14NecYesterday} $\dfrac\varphi {\y\varphi}$
\end{enumerate}
\begin{enumerate}[$\textsc{nec}_\nec$]
\item\label{ax14NecBox}  $\dfrac\varphi {\nec\varphi}$
\end{enumerate}
\begin{enumerate}[$\textsc{nec}_\has$]
\item\label{ax14NecBox'}  $\dfrac\varphi {\has\varphi}$
\end{enumerate}
\end{multicols}
\end{enumerate}
\end{definition}

Axiom group \ref{ax01Taut} concerns the relationship between $\imp$ and $\dimp$. In particular, Axiom \ref{axco01} is used in C.~Rauszer's axiomatisation of intuitionistic logic with co-implication (called Heyting--Brouwer logic)~\cite{Rauszer74}. The G\"odel--Dummett axiom~\ref{axGodel} and its order dual \ref{axcoGodel} reflect the fact that the connectives $\imp$ and $\dimp$ are implemented on locally linear posets (which conversely is \emph{necessary} for these axioms to be valid).

Axioms \ref{ax02Bot} to \ref{ax12:ind:2} concern the future. Axioms \ref{ax03NexWedge}, \ref{ax05KNext}, and \ref{ax06KBox} are standard modal axioms (viewing $\nx$ as a box-type modality). In particular they hold in any normal modal logic, although of course $\gtl$ is not itself normal by virtue of being strictly sub-classical. Axiom \ref{ax07:K:Dual} is an order dual version of \ref{ax06KBox}; such dual axioms are often needed in intuitionistic modal logic, since $\ps$ and $\nec$ are not typically interdefinable.
The axioms \ref{ax02Bot} and \ref{ax04NexVee} have to do with the passage of time being deterministic in linear temporal logic:  \ref{ax02Bot} characterises seriality and  \ref{ax04NexVee} characterises (partial) functionality, thus together they constrain temporal accessibility to be a total function. 

The co-inductive axiom \ref{ax09BoxFix} states that if something will henceforth be the case, then it is true now and, moreover, in the next moment, it will still henceforth be the case. Axiom \ref{ax11:ind:1} is successor induction, as future moments are indexed by the natural numbers with their usual ordering. Axioms \ref{ax10DiamFix} and \ref{ax12:ind:2} are analogues of \ref{ax09BoxFix} and \ref{ax11:ind:1} respectively. 
Axioms \ref{ax02Bot'} to \ref{ax12:ind:2'} are the past analogues of axioms \ref{ax02Bot} to \ref{ax12:ind:2}. The connection axioms~\ref{temporal1} and~\ref{temporal2} relate past and future tenses by combining the connectives $\nx$ and $\y$.

All rules of group \ref{group4} are standard modal logic deduction rules, and in particular any normal modal logic is closed under these rules.

Most of the axioms are either included in the axiomatisation of intuitionistic $\ltl$ \cite{Boudou2017} or a variant of one of them (e.g.~a contrapositive). 
From this, we easily derive the following.

\begin{proposition}
The above calculus is sound for the class of real-valued models, as well as for the class of bi-relational models.
\end{proposition}

\begin{proof}
The rules \ref{axDimpMon} and \ref{axDimpDis} are readily seen to preserve validity.
We check Axiom \ref{axcoGodel}; 
all other rules or axioms have been shown to be sound for intuitionistic or bi-relational models in the literature (see e.g.~\cite{Rauszer74,Balbiani2017}). 

\ref{axcoGodel}: Let us assume towards a contradiction that $\neg (( \varphi \dimp \psi) \wedge (\psi \dimp \varphi))$ is not valid with respect to bi-relational models, so we can find $\mathcal M=(W,T,{\leq},S,\allowbreak\val{\,\cdot\,})$ and $(w,t)\in\mathcal{M}$ such that $(w,t)\not \in \val{\neg (( \varphi \dimp \psi) \wedge (\psi \dimp \varphi))}$.
Therefore there exists $(v,t) \le (w,t)$ such that $(v,t) \in \val{( \varphi \dimp \psi) \wedge (\psi \dimp \varphi)}$.
Then $(v,t) \in \val{\varphi \dimp \psi}$ and  $(v,t) \in \val{\varphi \dimp \psi}$.
Hence there exist $(v',t)\ge (v,t)$ and $(v'',t)\ge (v,t)$ such that $(v',t) \in \val{\varphi}\setminus \val{\psi}$ and $(v'',t) \in \val{\psi}\setminus\val{\varphi}$. 
Since $(W,\le)$ is a linear order, either $(v',t) \le (v'',t)$ or $(v',t) \ge (v'',t)$.
In the former case we get that $(v',t) \in \val{\psi}$ and in the latter case we get that $(v'',t) \in \val{\varphi}$; in any case we reach a contradiction.\qedhere

\end{proof}

Our main objective is to show that our calculus is indeed complete; proving this will take up the remainder of this article.

As we show next, we can also derive the converses of some of these axioms. Below, for a set of formulas $\Gamma$ we define $\nx \Gamma = \{\nx\varphi \mid \varphi \in \Gamma\}$, and empty conjunctions and disjunctions are defined by $\bigwedge\varnothing =\top$ and $\bigvee \varnothing = \bot$.

When reasoning about $\gtl$, it is useful to note that by \ref{ax00Intu} and \ref{ax13MP} (modus ponens), the following weak form of the deduction theorem holds: for $\varphi, \psi \in \lanfull$, we have $\varphi \imp \psi \in \gtl$ if (and only if) we can deduce $\psi$ from $\varphi$ in the system with $\gtl$ as its axioms and modus ponens as its only deduction rule. We will often implicitly use this to simplify reasoning.

\begin{lemma}\label{lemmReverse}
	Let $\varphi \in \lanfull$ and $\Gamma\subseteq \lanfull$ be finite. Then the following formulas belong to $\gtl$.
	\begin{multicols}2
	\begin{enumerate}[label=(\alph*)]
		
		\item\label{itVee} $\nx (\bigvee \Gamma) \iiff \bigvee \nx \Gamma$
		\item\label{itWedge} $\nx (\bigwedge \Gamma) \iiff \bigwedge \nx \Gamma$
		\item\label{itUtoF} $\varphi\until\psi\imp\ps \psi$
		\item\label{itReverseDiam} $\varphi\until\psi \imp \psi\vee (\varphi \wedge \nx(\varphi \until \psi))$
		\item\label{itReverseBox} $\varphi \wedge \nx \nec \varphi \imp \nec \varphi$
		\item \label{itNotRef} $ (\varphi\dimp\varphi) \imp  \psi$
		\item\label{itVeePast} $\y (\bigvee \Gamma) \iiff \bigvee \y \Gamma$
		\item\label{itWedgePast} $\y (\bigwedge \Gamma) \iiff \bigwedge \y \Gamma$
		\item\label{itStoP} $\varphi\since\psi\imp\past \psi$
		\item\label{itReversePast} $\varphi\since\psi \imp \psi\vee (\varphi \wedge \y(\varphi \since \psi))$
		\item\label{itReverseHas} $\varphi \wedge \y \has \varphi \imp \has \varphi$
		\item\label{itReverseDimp} $(\varphi \dimp \psi) \imp \varphi$
	\end{enumerate}
\end{multicols}	
\end{lemma}

\begin{proof}
For claims \ref{itVee} and \ref{itWedge}, one direction is obtained from repeated use of Axioms \ref{ax04NexVee} or \ref{ax03NexWedge} and the other is proven using \ref{ax14NecCirc} and \ref{ax05KNext}; note that the first claim requires \ref{ax02Bot} to treat the case when $\Gamma = \varnothing$. Details are left to the reader.
For \ref{itUtoF}, we recall that $\ps\psi$ is shorthand for $\top\until\psi$.
Then one uses \ref{ax08:K:Dual} and the tautology $ \varphi\to \top$ to obtain $\varphi\until \psi\to \top\until\psi$.

For \ref{itReverseDiam}, let $\theta =  \psi \vee (\varphi\wedge \nx(\psi\until\varphi))$.
We claim that $\varphi\wedge\nx \theta\imp\theta$ is derivable, for which it suffices to show $ \nx \theta\imp \nx(\varphi\until\psi)$.
Using \ref{itVee} and \ref{itWedge}, this amounts to showing
\[ \nx\psi \vee ( \nx \varphi \wedge \nx \nx(\varphi\until\psi))  \imp \nx(\varphi\until\psi),\]
which would follow from both $   \nx\psi  \imp \nx (\varphi\until\psi)$ and $   \nx \varphi \wedge \nx \nx (\varphi\until\psi)  \imp \nx (\varphi\until\psi)$ via propositional reasoning.
But these follow from \ref{ax10DiamFix}, which yields $\psi\imp (\varphi\until\psi)$ and $\varphi\wedge \nx (\varphi\until\psi)  \imp  (\varphi\until\psi)$, along with routine reasoning involving $\nx$.

From $\varphi\wedge\nx \theta\imp\theta$, we may use necessitation and \ref{ax12:ind:2} to obtain $\varphi \until \theta \imp \theta$.
Now we note that $\psi\imp\theta$ is a tautology; hence using \ref{ax07:K:Dual}, we obtain $\varphi \until \psi \imp \theta$, i.e.~$\varphi \until \psi \imp \psi\vee(\varphi\wedge \nx (\varphi\until \psi))$, as needed.

Claim \ref{itReverseBox} is similar. Note that $\nec \varphi \imp \varphi$ holds by \ref{ax09BoxFix}, so we have $\nx\nec \varphi \imp \nx \varphi$ by \ref{ax14NecCirc}, \ref{ax05KNext}, and modus ponens as before.
Similarly, $\nec \varphi \imp \nx \nec \varphi$ holds by \ref{ax09BoxFix}, so \mbox{$\nx \nec \varphi \imp \nx \nx\nec \varphi$} holds by \ref{ax14NecCirc}, \ref{ax05KNext}, and modus ponens. 
Hence we have 
\[\nx \nec \varphi \imp \nx \varphi \wedge \nx \nx \nec \varphi.\]
Using \ref{ax03NexWedge} and some propositional reasoning we obtain
\[\varphi \wedge \nx \nec \varphi \imp \nx (\varphi \wedge \nx \nec \varphi ).\]
By $\nec$-necessitation and \ref{ax11:ind:1}, we have \[\varphi \wedge \nx \nec \varphi \imp \nec (\varphi \wedge \nx \nec \varphi ).\] Since
$ \nec (\varphi \wedge \nx \nec \varphi)\imp \nec\varphi $ can be proven using \ref{ax06KBox}, we obtain $ \varphi \wedge \nx \nec \varphi \imp \nec\varphi$, as needed.

Claim~\ref{itNotRef} can be derived from the intuitionistic tautology $\varphi\imp\psi\vee\varphi$ and Rule~\ref{axDimpDis}.

The proofs of claims~\ref{itVeePast}--\ref{itReverseDimp} follow similar lines of reasoning as for their future versions.
\end{proof}


\section{The canonical system and canonical quasimodel}\label{secCan}

In this section we first construct a standard canonical structure for $\gtl$.
In the presence of $\until$, $\since$, $\nec$, and $\has$, this canonical structure is only a labelled system, rather than a proper bi-relational model.
Nevertheless, it will be a useful ingredient in our completeness proof, for we can apply the quotient construction of \Cref{Sec:quotient} to it and show that the resulting labelled system is in fact a quasimodel.

Since we are working over an intermediate logic, the role of maximal consistent sets will be played by complete types, as defined below (see \Cref{definition:type} for the definition of {\em type}).
Below, recall that by convention, $\bigwedge\varnothing = \top$ and $\bigvee\varnothing=\bot$.

\begin{definition}\label{def:complete}
Given two sets of formulas $\Gamma,\Delta \subseteq\lanfull$, we say that $\Delta$ is a \define{consequence} of $\Gamma$, denoted by $\Gamma \vdash \Delta$, if there exist finite (possibly empty) $\Gamma'\subseteq \Gamma$ and $\Delta' \subseteq \Delta$ such that $\bigwedge \Gamma' \imp \bigvee \Delta' \in \gtl$.
We say that the pair $(\Gamma, \Delta)$ is \define{consistent} if $\Gamma \not\vdash \Delta$. 
We will call a $\Sigma$-type $\Phi$ consistent if $(\Phi, \Sigma \setminus \Phi)$ is consistent.
\end{definition}

Note that we are using the standard interpretation of $\Gamma \vdash \Delta$ in Gentzen-style calculi. When working within a turnstile, we will follow the usual proof-theoretic conventions of writing $\Gamma,\Delta$ instead of $\Gamma \cup \Delta$, and writing $\varphi$ instead of $\{\varphi\}$.

We will often want to think of a consistent $\Sigma$-type $\Phi$ as representing the pair $(\Phi, \Sigma \setminus \Phi)$. Note that for such a $\Phi$, the set $\Sigma \setminus \Phi$ cannot contain any formula in $\gtl$ since $\varphi \in \gtl$ implies $\varnothing \vdash \varphi$.

\begin{lemma}\label{lemmcompleteIsType}
If $\Sigma \subseteq \lanfull$ is closed under subformulas, $\Phi \subseteq \Sigma$, and $(\Phi, \Sigma \setminus \Phi)$ is consistent, then $\Phi$ is a (consistent) $\Sigma$-type.
\end{lemma}

\begin{proof}
We check the conditions of \Cref{definition:type}.

\begin{itemize}
\item Condition~\ref{type1}: If $\varphi \wedge \psi \in \Phi$ then since $\varphi \wedge \psi \imp \varphi$ is an intuitionistic tautology, we cannot have $\varphi \in \Sigma \setminus \Phi$; thus $\varphi \in \Phi$. Similarly, $\psi \in \Phi$. The converse is even more trivial, using the tautology $\varphi \wedge \psi \imp \varphi \wedge \psi$.

\item Condition~\ref{type2}: If $\varphi \vee \psi \in \Phi$ then since $\varphi \vee \psi \imp \varphi \vee \psi$ is an intuitionistic tautology, we cannot have both $\varphi \in \Sigma \setminus \Phi$ and $\varphi \in \Sigma \setminus \Phi$; thus either $\varphi \in \Phi$ or $\varphi \in \Phi$.  The converse uses the tautologies $\varphi \imp \varphi \vee \psi$ and $\psi \imp \varphi \vee \psi$.

	\item Condition~\ref{type3}: (a) If $\varphi \imp \psi \in \Phi$ and $\varphi  \in \Phi$, then since $(\varphi \imp \psi) \wedge \varphi \imp \psi$ is an intuitionistic tautology, we cannot have $\psi \in \Sigma \setminus \Phi$; thus $\psi \in \Phi$.
	
	(b) If $\psi  \in \Phi$ then since $\psi \imp (\varphi  \imp \psi)$ is an intuitionistic tautology, we cannot have $\varphi \imp \psi \in \Sigma \setminus \Phi$; thus $\varphi \imp \psi \in   \Phi$.

	\item Condition~\ref{type4}: (a) By \Cref{lemmReverse}\ref{itReverseDimp}, we know that $\varphi \dimp \psi  \vdash \varphi$. Hence if $\varphi \dimp \psi \in \Phi$ we cannot have $\varphi \in \Sigma \setminus \Phi$; thus $\varphi \in  \Phi$.
	
	 (b)  By Axiom~\ref{axco01} we know that $\varphi \vdash \psi \vee (\varphi \dimp \psi)$. Hence if $\varphi \in \Phi$ and $\psi \in \Sigma \setminus \Phi$, we cannot have $\varphi \dimp \psi \in \Sigma \setminus \Phi$; thus $\varphi \dimp \psi \in \Phi$.\qedhere
\end{itemize}
\end{proof}
 As with maximal consistent sets, consistent types satisfy a Lindenbaum property.
Below, if $(\Gamma,\Delta)$ and $(\Gamma',\Delta')$ are pairs of sets of formulas, we say that $(\Gamma',\Delta')$ \textbf{extends} $(\Gamma,\Delta)$ if $\Gamma\subseteq\Gamma'$ and $\Delta\subseteq\Delta'$. A consistent $\Sigma$-type $\Phi$ extends $(\Gamma, \Delta)$ if $(\Phi, \Sigma \setminus \Phi)$ extends $(\Gamma, \Delta)$.

\begin{lemma}[Lindenbaum lemma]\label{LemmLind}
Let $\Sigma \subseteq \lanfull$ be closed under subformulas, and $\Gamma,\Delta \subseteq \Sigma$. If $\Gamma \not\vdash \Delta$, then there exists a consistent $\Sigma$-type extending $(\Gamma,\Delta)$. 
\end{lemma}

\proof
The proof is standard, but we provide a sketch.
Let $\varphi \in \Sigma$. Note that either $\Gamma ,\varphi  \not\vdash \Delta $ or $\Gamma  \not \vdash \Delta,\varphi$, for otherwise by an application of the cut rule (which is intuitionistically derivable) we would have $\Gamma \vdash \Delta$. Thus we can add $\varphi$ to $\Gamma$ or to $\Delta$ and remain consistent. By repeating this process for each of the countable number of formulas of $\Sigma$ (or using Zorn's lemma) we can find a consistent $(\Gamma', \Delta')$ extending $(\Gamma, \Delta)$ such that $\Gamma' \cup \Delta' = \Sigma$. Since $(\Gamma', \Delta')$ is consistent,  $\Delta'$ is also disjoint from $\Gamma'$; hence $\Delta' = \Sigma \setminus \Gamma'$. Thus by \Cref{lemmcompleteIsType}, $\Gamma'$ is a consistent $\Sigma$-type.
\endproof

Before defining the canonical system, recall that for a set of formulas $\Gamma$, we have $\nx \Gamma \eqdef \lbrace \nx \varphi \mid   \varphi \in \Gamma\rbrace$.
We also define
\[\circop\Gamma \eqdef \lbrace  \varphi \mid \nx \varphi \in \Gamma\rbrace.\]

Given a set $A$, let $\mathbb I_A$ denote the identity function on $A$.
The canonical system $\CMod$ is defined as the labelled structure
\[
\CMod = (\ptype{},{\peq_\CIcon },\ell_\CIcon,S_\CIcon  ),
\]
where \begin{itemize}
\item
$\ptype{}$ is the set of consistent $\lanfull$-types,
\item
$\Phi \peq_\CIcon \Psi$ if $\Phi \supseteq \Psi$,
 \item
 $\ell_\CIcon(\Phi)=\Phi$,
\item
 $S_\CIcon(\Phi) = \circop \Phi$.
 \end{itemize}
We will usually omit writing $\ell_\CIcon $, as it has no effect on its argument. We will verify shortly that $\circop \Phi$ is indeed a consistent $\lanfull$-type.

Next we show that $\CMod$ is an $\lanfull$-labelled system.
We begin by showing that it is based on a labelled space.
\begin{lemma}
	\label{lemm:normality} $(\ptype{},\peq_\CIcon,\ell_\CIcon)$ is an $\lanfull$-labelled space.
\end{lemma}

\begin{proof}
	It is trivial that $\peq_\CIcon $ is a partial order. Moreover, $\ell_\CIcon $ is the identity, so $\Phi \peq_\CIcon  \Psi$ implies that $\ell_\CIcon  (\Phi) \supseteq \ell_\CIcon  (\Psi)$; that is, $\ell_\CIcon$ is inversely monotone.
	
	To prove that $(\ptype{},{\peq_\CIcon })$ is locally linear, assume towards a contradiction that it is not. We consider two cases:
	\begin{enumerate}[wide, labelwidth=!, labelindent=0pt]
		\item There exist $\Phi$, $\Psi$ and $\Theta$
		such that $\Phi \subseteq \Psi$ and $\Phi \subseteq \Theta$, but $\Psi \not  \subseteq \Theta$ and $\Theta \not  \subseteq \Psi$. Thus there exist two formulas $\varphi \in \Theta \setminus \Psi$ and $\psi \in \Psi \setminus \Theta$. 
		Then it is easy to see that $\varphi \imp \psi \not \in \Theta$ and $\psi \imp \varphi \not \in \Psi$. Thus neither of these formulas is in $\Phi$, so by condition~\ref{type2} of \Cref{definition:type} their disjunction, Axiom~\ref{axGodel}, is in $\lanfull \setminus \Phi$---a contradiction.
		
		\item There exist $\Phi$, $\Psi$ and $\Theta$
		such that $\Phi \supseteq \Psi$ and $\Phi \supseteq \Theta$, but $\Psi \not  \supseteq \Theta$ and $\Theta \not  \supseteq \Psi$. Then it is easy to see that there exist two formulas $\varphi \in \Psi \setminus \Theta$ and $\psi \in \Theta \setminus \Psi$ such that $\varphi \dimp \psi\in \Psi$ and $\psi \dimp \varphi \in \Theta$. From $\Phi \supseteq \Psi$ and $\Phi \supseteq \Theta$, we conclude that $ \varphi \dimp \psi, \psi \dimp \varphi \in \Phi$. But this contradicts the consistency of $\Phi$, since by Axiom~\ref{axcoGodel} we have $ \varphi \dimp \psi ,\psi \dimp \varphi \vdash \bot$.
	\end{enumerate}
	We finish by considering the conditions on $\imp$ and $\dimp$ in the definition of a labelled space. Let us consider $\Phi \in \ptype{}$.
	\begin{itemize}[wide, labelwidth=!, labelindent=0pt]
		\item Suppose $\varphi \imp \psi \not\in \Phi$. Then by condition~\ref{type3b} of \Cref{definition:type}, $\psi \not\in \Phi$. Let us define $u = (\Phi \cup \lbrace \varphi \rbrace, \lbrace \psi \rbrace)$, and let us assume for a contradiction that $u$ is inconsistent. This means that there exists $\gamma\in \Phi$ such that $\gamma \wedge \varphi \imp \psi \in \gtl$.
		By propositional reasoning, $\gamma\imp (\varphi \imp \psi) \in \gtl$. Since $\gamma \in \Phi$ and $\Phi$ is consistent, $\varphi \imp \psi \not\in \lanfull \setminus \Phi$, that is, $\varphi \imp \psi \in  \Phi$---a contradiction. Hence $u$ is consistent and by \Cref{LemmLind}  can be extended to a consistent $\lanfull$-type $\Psi$. From the definition of $u$ we can conclude that $\Psi \leq_\CIcon \Phi$, $\varphi \in \Psi$, and $\psi \not\in \Psi$ as required.

		\item Suppose $\varphi \dimp \psi \in \Phi$. Then by condition~\ref{type4a} of \Cref{definition:type}, $\varphi \in \Phi$. Let us define $u = (\lbrace \varphi\rbrace, (\lanfull \setminus \Phi)\cup \lbrace \psi\rbrace)$, and let us assume by contradiction that $u$ is inconsistent. This means that there exists $\gamma \in \lanfull \setminus \Phi$ such that $\varphi \imp \psi \vee \gamma \in \gtl$. By Rule \ref{axDimpDis}, we get  $(\varphi \dimp \psi)\imp \gamma \in \gtl$. Since $\gamma \in \lanfull \setminus \Phi$, we deduce that $\varphi \dimp \psi \in \lanfull \setminus \Phi$---a contradiction. By \Cref{LemmLind}, $u$ can be extended to a complete type $\Psi$. It is easy to check that $\Phi \leq_\CIcon \Psi$, $\varphi \in \Psi$, and $\psi \not\in \Psi$ as required.
	\qedhere\end{itemize}
\end{proof}

\begin{lemma}\label{lem:welldefined}
	The function $S_\CIcon \colon \ptype{}\to  \ptype{} $ is well defined. That is, $S_\CIcon( \Phi) \in \ptype{}$ for every $\Phi \in \ptype{}$.
\end{lemma}

\begin{proof}
	Let $\Phi\in \ptype{} $ and $\Psi = S_\CIcon(\Phi) = \circop \Psi$.
	By \Cref{lemmcompleteIsType}, we only need to check that $(\Psi, \lanfull \setminus \Psi)$ is consistent. Suppose not; then let $\Gamma \subseteq \Psi$ and $\Delta\subseteq \lanfull \setminus \Psi$ be finite and such that $\bigwedge \Gamma \imp \bigvee \Delta \in \gtl$. Using \ref{ax14NecCirc} and \ref{ax05KNext} we see that $\nx \bigwedge \Gamma \imp  \nx \bigvee \Delta \in \gtl$, which in view of \Cref{lemmReverse} implies that $ \bigwedge \nx \Gamma \imp  \bigvee  \nx \Delta \in \gtl$ as well. But $\nx \Gamma \subseteq \Phi$ and $\nx \Delta \subseteq \lanfull \setminus \Phi$ (since $\delta \in \Delta$ implies $\delta \not\in \circop \Phi$, which implies $\nx \delta \not\in \Phi$), contradicting the fact that $\Phi$ is consistent. We conclude that $\Psi \in \ptype{}$.
\end{proof}

In order to see that the function $S_\CIcon$ is a fully confluent, convex relation, it will be useful to note that $S_\CIcon$ is a bijection.

\begin{lemma}\label{lemma:bijection}
The function $S_\CIcon$ is a bijection.
\end{lemma}

\begin{proof}
Define $\yop\Phi \eqdef \lbrace  \varphi \mid \y \varphi \in \Phi\rbrace$. We claim that $\Phi \mapsto \yop\Phi$ is the inverse of $S_\CIcon$. By Axiom~\ref{temporal1}
\[\varphi \in \yop\circop\Phi \iff \nx\y\varphi \in \Phi \iff \varphi \in \Phi,\]
and by Axiom~\ref{temporal2}
\[\varphi \in \circop\yop\Phi \iff \y\nx\varphi \in \Phi \iff \varphi \in \Phi.\qedhere\]
\end{proof}

\begin{lemma}\label{lem:confluent}
	The function $S_\CIcon$ is a fully confluent relation.
\end{lemma}
\begin{proof} We check the four confluence conditions.
	\begin{description}[wide, labelwidth=!, labelindent=0pt ]

		\item[\textbf{Forth--down, forth--up:}] Let $\Phi$ and $\Psi$ be such that $\Phi \supseteq \Psi $.
Since $S_\CIcon$ is a function, these two confluence properties amount to showing that $S_\CIcon(\Phi) \supseteq S_\CIcon(\Psi)$.
If $\varphi\in S_\CIcon(\Psi)  $ then $\nx\varphi\in \Psi$, which since $\Phi \supseteq \Psi $, implies that $\nx \varphi\in \Phi$, and hence $\varphi\in S_\CIcon(\Phi) $.
		
		
		\item[\textbf{Back--down, back--up: }] Let $\Phi$ and $\Psi$ be such that $\Phi \supseteq \Psi $.
Since, by \Cref{lemma:bijection}, $S_\CIcon^{-1}$  is a function, these two confluence properties amount to showing that $S_\CIcon^{-1}(\Phi) \supseteq S_\CIcon^{-1}(\Psi)$, that is, $\y^{-1}(\Phi) \supseteq \y^{-1}(\Psi)$.
If $\varphi\in \y^{-1}(\Psi)  $ then $\y\varphi\in \Psi$, which since $\Phi \supseteq \Psi $, implies that $\y \varphi\in \Phi$, and hence $\varphi\in \y^{-1}(\Phi) $.\qedhere
	\end{description}
\end{proof}

%

\begin{lemma}\label{lem:convex}
The function $S_\CIcon$ is a convex relation.	
\end{lemma}
\begin{proof}
Since $S_\CIcon$ is a function, images of points are singletons and hence automatically convex. The same is true for preimages of points, since by \Cref{lemma:bijection}, $S_\CIcon$ is an injection.
\end{proof}

\begin{lemma}\label{lem:sensible}
	The function $S_\CIcon$ is a sensible relation.
\end{lemma}
\begin{proof}

Let us consider $\Phi$ and $\Psi$ such that $\Psi=S_\CMod(\Phi)$. We verify the conditions required for $(\Phi,\Psi)$ to be sensible (\Cref{compatible}).

We have $\varphi \in \Psi$ if and only if $\nx\varphi \in \Phi$ by the definition of $S_\CIcon$. Since $\Phi = \y^{-1}(\Psi)$, we similarly have $\varphi \in \Phi$ if and only if $\y\varphi \in \Psi$

If $\nec \varphi \in \Phi$ then, by Axiom~\ref{ax09BoxFix} we get $\varphi, \nx \nec \varphi \in \Phi$. Since $\Psi = S_\CIcon(\Phi)$, we get $\nec \varphi \in \Psi$.
Conversely, assume that $\nec \varphi \in \Phi^-$. By \Cref{lemmReverse}\ref{itReverseBox}, $\varphi \wedge \nx \nec \varphi \in \Phi^-$, so either $\varphi \in \Phi^-$ or $\nx \nec \varphi \in \Phi^-$ (giving in the second case $\nec \varphi \in \Psi^-$). In either case we reach the desired conclusion.

For the case of $\until$, by \Cref{lemmReverse}\ref{itReverseDiam}, we know $\varphi\until\psi \vdash\psi, \varphi \wedge \nx(\varphi \until \psi)$. Hence if $\varphi \until \psi \in \Phi$ then either  $\psi \in \Phi$ or $\varphi \wedge \nx (\varphi \until \psi) \in \Phi $. In the first case we are done immediately. In the second, since $\Phi$ is an $\lanfull$-type we deduce $\varphi \in \Phi$ and $\nx(\varphi \until \psi) \in \Phi$, the latter of which gives $\varphi \until \psi \in \Psi$, and we are done.  Conversely, by Axiom~\ref{ax10DiamFix} and some intuitionistic reasoning, both $\psi\vdash \varphi \until \psi$ and $\varphi , \nx(\varphi \until \psi) \vdash \varphi \until \psi$. From the first, we see that $\psi \in \Phi$ implies $\varphi \until \psi \in \Phi$. If, alternatively, both $\varphi \in \Phi$ and $\varphi \until \psi \in \Psi$, then $\varphi \in \Phi$ and $\nx(\varphi \until \psi) \in \Phi$, which together with $\varphi , \nx(\varphi \until \psi) \vdash \varphi \until \psi$, also yields $\varphi \until\psi \in \Phi$.

The cases for  $\has$ and $\since$ are the temporal converses of those for $\nec$ and $\until$ respectively (using  again the fact that $S_\CIcon^{-1}$ is to $\y$ as $S_\CIcon$ is to $\nx$).
\end{proof}

We note the general fact that given a $\Sigma_1$-labelled system and a subformula-closed $\Sigma_2 \subseteq \Sigma_1$, one can restrict the labelling to $\Sigma_2$ in the natural way (by replacing its value at any point by its intersection with $\Sigma_2$).
Doing so yields a $\Sigma_2$-labelled system. This is easily verifiable from the definitions.

\begin{proposition}\label{prop:CisW}
The canonical system $\CMod$ is an $\lanfull$-labelled system. Restricting the labelling to any subformula-closed $\Sigma \subseteq \lanfull$ yields a $\Sigma$-labelled system.
\end{proposition}

\proof
For the first claim, we need for the following three properties to hold:
\begin{enumerate}
\item $(\ptype{},{\peq}_\CIcon ,\ell_\CIcon )$ is a labelled space;
\item $S_\CIcon $ is a bi-serial, fully confluent, convex, sensible relation; and
\item $\ell_\CIcon $ has $\type{\lanfull}$ as its codomain.
\end{enumerate}
The first item is \Cref{lemm:normality}.  The relation $S_\CIcon $ is bi-serial since it is a well defined bijection, and it is a fully confluent, convex, sensible relation by Lemmas \ref{lem:confluent},~\ref{lem:convex}, and \ref{lem:sensible}.
 Finally, if $\Phi \in \ptype{}$ then $\ell_\CIcon (\Phi) = \Phi$, which is an element of $\type{\lanfull}$ by \Cref{lemmcompleteIsType}.
 
The second claim follows from the observation preceding the proposition.
\endproof

By \Cref{prop:quotient_system}, the structure $\nicefrac{\CMod}\Sigma \eqdef (\ptype{} / {\sim}, \leq_\mathcal Q, R^+_\mathcal Q, \ell_\mathcal Q)$ formed with the quotient construction of \Cref{Sec:quotient} is a $\Sigma$-labelled system. We call $\nicefrac{\CMod}\Sigma$ the \define{canonical quasimodel}, although we do not yet know it is a quasimodel, since we do not yet know that $R^+_\mathcal Q$ is $\om$-sensible.
In the sequel, we show that this is the case.
From this we will obtain completeness of our deductive calculus, since any formula that is not derivable is falsified in the canonical quasimodel and hence on some model by unwinding.


\section{Characteristic formulas}\label{SecChar}

 In this section, we show that there exist formulas defining points in the canonical quasi\-model $\nicefrac{\CMod}\Sigma = (\ptype{} / {\sim}, \leq_\mathcal Q, R^+_\mathcal Q, \ell_\mathcal Q)$, i.e.~to each $w \in \ptype{} / {\sim}$ we assign two formulas that together `distinguish' $w$. In the next section, these formulas will play an essential role in the proof that $R^+_\mathcal Q$ is $\om$-sensible, and hence that $\nicefrac{\CMod}\Sigma$ is a quasimodel.

First, we define a formula $\chi^+_\Sigma(w)$ (or $\chi^+(w)$ when $\Sigma$ is clear from context) such that for all $\Gamma \in \ptype{}$, $\chi^+(w) \in \Gamma$ if and only if   $[\Gamma] \leq w$.
Dually, we define $\chi^-(w)$ so that for all $\Gamma \in \ptype{}$, $\chi^-(w) \notin \Gamma$ if and only if  $w\leq[\Gamma]$.
Compared to \cite{eventually}, these formulas require co-implication, as they must look `up' as well as `down' the model.
In this section, we write $\cqm\Sigma = (\ptype{} / {\sim},\leq,R,\ell)$.
We will omit subscript indices on the $\ell$ and $L$ functions.

\begin{definition}
Fix a $\Sigma \subseteq\lanfull$ that is finite and closed under subformulas.
Given $\Delta\in \type\Sigma$, define $\overrightarrow \Delta \eqdef \bigwedge\Delta\imp \bigvee(\Sigma \setminus \Delta)$ and $\overleftarrow \Delta \eqdef \bigwedge\Delta\dimp \bigvee(\Sigma \setminus \Delta)$.
Given $w  \in \ptype{} / {\sim}$, we define a formula $\chi^0(w)$ by
\[\chi^0(w) \eqdef    \bigwedge_{\Delta \in L(w)} {\dneg} \overrightarrow{\Delta} \wedge  \bigwedge_{\Delta \in \type\Sigma \setminus L(w)} \neg \overleftarrow {\Delta}  . \]
Then define $\chi^+ (w)$ by
\[\chi^+(w) \eqdef \overleftarrow {\ell(w)} \wedge \chi^0(w) \]
and $\chi^- (w)$ by
\[\chi^-(w) \eqdef \chi^0(w) \Rightarrow \overrightarrow{\ell(w)}.\]
\end{definition}

\begin{proposition}\label{propSimForm}
Given $w \in \ptype{} / {\sim}$ and $\Gamma\in \ptype{}$,
\begin{enumerate}[label=\arabic*)]

\item\label{simulability:c0} $\chi^0(w)\in \Gamma $ if and only if  $L (\Gamma) = L(w)$,

\item\label{simulability:c1} $\chi^+(w)\in \Gamma$ if and only if  $[\Gamma] \leq w$, and

\item\label{simulability:c2} $\chi^-(w)\not\in \Gamma $ if and only if  $[\Gamma] \geq w$.

\end{enumerate}
\end{proposition}

\proof
Let $w \in \ptype{} / {\sim}$ and $\Gamma\in \ptype{}$.

\medskip

\noindent \ref{simulability:c0}
For the left-to-right implication, assume that  $\chi^0(w)\in \Gamma$, so that both $ \bigwedge_{\Delta \in L(w)}{\dneg} \overrightarrow \Delta  \in \Gamma$ and $   \bigwedge_{\Delta \in \type\Sigma \setminus   L(w)}\neg \overleftarrow {\Delta}  \in \Gamma$, by the defining conditions of $\lanfull$-types.
For the inclusion $L(w) \subseteq L(\Gamma)$, let $\Delta\in L(w)$.
From ${\dneg} \overrightarrow \Delta \in \Gamma$, we obtain $\Phi \geq \Gamma$ such that $ \overrightarrow {\Delta} \notin \Phi$. Hence there is $\Phi_\Delta\leq \Phi $ with $\bigwedge \Delta\in \Phi_\Delta$ and $\bigvee (\Sigma \setminus \Delta)\not\in \Phi_\Delta$; i.e.~$\ell (\Phi_\Delta)=\Delta$.
From local linearity we see that $\Phi_\Delta\compa \Gamma$; hence $\Delta=\ell (\Phi_\Delta) \in L (\Gamma)$.

For the inclusion $L(w) \supseteq L(\Gamma)$, let $\Delta \in \type\Sigma\setminus L(w)$. Then for any $\Psi\leq \Gamma$ we have that $\overleftarrow{\Delta} \notin \Psi $, so that there is no $\Psi_\Delta\geq \Psi$ with $\bigwedge\Delta \in \Psi_\Delta$ and $\bigvee (\Sigma \setminus \Delta)\not\in \Psi_\Delta$.
Thus there is no $\Psi_\Delta \compa \Gamma$ with $\bigwedge\Delta \in \Psi_\Delta$ and $\bigvee (\Sigma \setminus \Delta)\not\in \Psi_\Delta$ (for the $\Psi_\Delta \geq \Gamma$ case, set $\Psi = \Gamma$; for $\Psi_\Delta \leq \Gamma$ set $\Psi = \Psi_\Delta$). We deduce that $\Delta\notin L (\Gamma)$.

For the right-to-left implication, assume that $L (\Gamma) = L(w) $.
Then for $\Delta\in L (\Gamma)$ we readily obtain ${\dneg}\overrightarrow{\Delta} \in \Gamma$, and similarly for $\Delta \in \type\Sigma \setminus L (\Gamma)$ we obtain $\neg \overleftarrow{\Delta} \in \Gamma$, from which we obtain by propositional reasoning that $\chi^0(w) \in \Gamma$.

\medskip

\noindent\ref{simulability:c1}
If $\chi^+(w)\in \Gamma$ then $\chi^0(w) \in \Gamma$,  and so by item~(\ref{simulability:c0}, $L (\Gamma) = L(w)$, while $\overleftarrow{\ell(w)} \in \Gamma$ implies there is some $\Gamma'\geq\Gamma$ with $\ell (\Gamma') = \ell(w)$.
This shows that $w = [\Gamma'] \geq [\Gamma]$, as claimed.

Conversely, if $[\Gamma] \leq w$, then $L(\Gamma) = L(w)$, and so by item~(\ref{simulability:c0}, $\chi^0(w)\in \Gamma$. Further, there exists $\Gamma'$ such that $[\Gamma'] = w$ and $\Gamma' \geq \Gamma$, which implies $\overleftarrow {\ell(w)} = \overleftarrow {\ell([\Gamma'])} =   \overleftarrow {\ell(\Gamma')}\in \Gamma$ also.

\medskip

\noindent\ref{simulability:c2}
If $\chi^-(w)\not\in \Gamma $ then there exists $\Phi \leq \Gamma$ with $\chi^0(w) \in \Phi$ and $\overrightarrow{\ell(w)} \not\in \Phi$. From  $\chi^0(w) \in \Phi$ we deduce $L(\Gamma) = L(\Phi) = L(w)$. From $\overrightarrow{\ell(w)} \not\in \Phi$ we deduce there exists $\Gamma' \leq \Phi$ with $\ell(\Gamma') = \ell(w)$. Since also $L(\Gamma') = L(\Phi) = L(w)$, we have $[\Gamma'] = w$. Then since $\Gamma \geq \Phi \geq \Gamma'$ we have $[\Gamma] \geq [\Gamma'] = w$.

Conversely, if $[\Gamma] \geq w$, then there exists $\Gamma' \leq \Gamma$ with $[\Gamma'] = w$. Now we have $\chi^0(w)\in \Gamma'$ but $\overrightarrow{\ell(w)} \not \in \Gamma'$. It follows that $\chi^-(w)\not\in \Gamma $.
\endproof

\begin{remark}
Note that the formula $\chi^+_\Sigma(w)$ makes essential use of co-implication, as properties of $w\geq [\Gamma]$ do not affect truth values in $\Gamma$ in the language with $\imp$ alone.
In contrast, the formulas $\chi^-_\Sigma$ are similar to the formulas $\operatorname{Sim}(w)$ of \cite{eventually}, although we remark that co-implication is still needed to describe the full linear component of $w$.
\end{remark}

Next we establish some provable properties of each of $\chi^+_\Sigma$ and $\chi^-_\Sigma$.
We begin with the former.

\begin{proposition}\label{propsubplus}
Given $w \in \ptype{} / {\sim}$ and $\psi\in \Sigma $:
\begin{enumerate}[label=\arabic*)]
	\item\label{itPropsubplOne} If $\psi\in \Sigma \setminus \ell({{w}})$, then
$ \chi^+(w)\imp (\chi^+(w) \dimp \psi ) \in \gtl$.

	\item\label{itPropsubplOneb} If $\psi\in \ell({w})$, then
	$ \chi^+(w) \imp \psi \in \gtl$.

	

	\item\label{itPropsubplFive}
For any $w\in \ptype{} / {\sim}$,
$\displaystyle \chi^+(w) \imp \nx\bigvee _{{{w}} \mathrel R {{{v}}}  } \chi^+(v) \in \gtl$.
\end{enumerate}
\end{proposition}

\proof
\noindent \ref{itPropsubplOne}
Let $\psi\in \Sigma \setminus \ell(w)$. It suffices to show that for an arbitrary $\Gamma\in \ptype{}$, we have that  $\chi^+(w) \in \Gamma$ implies $ \chi^+(w) \dimp \psi\in \Gamma $, for this demonstrates, by \Cref{LemmLind}, that $(\chi^+(w), \chi^+(w) \dimp \psi )$ is inconsistent, giving $ \chi^+(w)\imp (\chi^+(w) \dimp \psi ) \in \gtl$.
From $\chi^+(w) \in \Gamma$ and \Cref{propSimForm} we obtain $\Gamma'\leq \Gamma$ such that $[\Gamma'] = w$, and hence $\chi^+(w) \in \Gamma'$ and $\psi\not\in\Gamma'$, yielding $\chi^+(w) \dimp \psi \in \Gamma $.

\medskip

\noindent \ref{itPropsubplOneb} With similar reasoning as for \ref{itPropsubplOne}, let $\psi\in \ell ({w})$ and $\Gamma\in \ptype{}$, and suppose $\chi^+(w) \in \Gamma$. Let $\Gamma'\leq \Gamma$ with $[\Gamma'] =w$.
Then $\psi\in\Gamma' $, yielding $\psi\in \Gamma$.

\medskip




\noindent \ref{itPropsubplFive} Let $\Gamma$ be such that $\chi^+(w) \in \Gamma$, so that there is $\Gamma'\geq \Gamma$ with $[\Gamma'] = w$.
Then $w\mathrel R [S_\CMod (\Gamma')]$ by definition, and moreover $ \chi^{+}([S_\CMod (\Gamma')]) \in  S_\CMod ( \Gamma')$, which implies that $\nx \chi^{+}([S_\CMod (\Gamma')]) \in \Gamma'$.
Thus $\nx \chi^{+}([S_\CMod (\Gamma')]) \in \Gamma$ since $\Gamma' \subseteq \Gamma$, so that $\nx\bigvee_{w\mathrel R  v} \chi^+(v) \in \Gamma$.
\endproof

The formula $\chi^-_\Sigma$ behaves `order dually', as established below.

\begin{proposition}\label{propsubminus}
Given $w \in \ptype{} / {\sim}$ and $\psi\in \Sigma $:
\begin{enumerate}[label=\arabic*)]
	\item\label{itPropsubOne} If $\psi\in \Sigma \setminus\ell({{w}})$, then
$ \psi\imp \chi^-(w) \in \gtl$.

	\item\label{itPropsubOneb} If $\psi\in \ell({w})$, then
$ \big (\psi \imp \chi^- (w) \big )\imp \chi^- (w) \in \gtl $.


	\item\label{itPropsubFive}
For any $w\in \ptype{} / {\sim}$,
$\displaystyle \nx(\bigwedge _{{{w}} \mathrel R  {{{v}}}  } \chi^-(v)) \imp \chi^-(w) \in \gtl$.

\end{enumerate}
\end{proposition}

\proof
\noindent \ref{itPropsubOne}
Assume that $\psi\in \Sigma \setminus \ell ({w})$ and $\psi \in \Gamma$ and write $w=[\Gamma']$.
Then $\psi\not\in \Gamma'$, which means $\Gamma \not\geq \Gamma'$, so $[\Gamma] \not\geq w$. Hence \Cref{propSimForm}(\ref{simulability:c2} implies that $\chi^-(w)\in \Gamma$.

\medskip

\noindent \ref{itPropsubOneb}
Suppose that $\psi\in \ell({w})$ and proceed to prove the claim by contraposition.
If $\chi^-(w) \not\in \Gamma$ for some $\Gamma\in \ptype{}$, then there exists $\Gamma' \leq \Gamma$ such that $w=[\Gamma' ]$.
But then $\chi^-(w)\not\in \Gamma' $ and $\psi\in \Gamma' $, which implies that $ \psi\imp\chi^-(w) \not \in \Gamma' $, and hence also $ \psi\imp\chi^-(w) \not\in \Gamma$, as required.

\medskip


\noindent \ref{itPropsubFive}
Proceed by contraposition.
If $\chi^-(w)\not\in \Gamma$ for some $\Gamma\in \ptype{}$, then there exists $\Gamma' \leq \Gamma$ such that $w=[\Gamma' ]$.
We have that $w \mathrel R [S_\CMod(\Gamma' )]$ by definition.
Letting $v= [S_\CMod(\Gamma' )]$, we have that $\chi^-(v) \not\in  S_\CMod(\Gamma' )$, and hence $\nx \chi^-(v) \not\in \Gamma' $. Then since $\Gamma' \supseteq \Gamma$, we have $\nx \chi^-(v) \not \in   \Gamma$.
Hence $\nx \bigwedge_{w\mathrel R v} \chi^-(v) \not \in \Gamma$.
\endproof

\section{Completeness}\label{SecComp}

In this section, we use the formulas $\chi^\pm_\Sigma$  to show that the relation $R_{\mathcal Q}^+$ on $\cqm\Sigma$ is $\om$-sensible and hence $\cqm\Sigma$ is a quasimodel.
Since validity over the class of quasimodels is equivalent to bi-relational validity, by Propositions~\ref{second} and \ref{falsifies}, and bi-relational validity is equivalent real-valued validity, by \Cref{thm:equal}, completeness will follow.

The following lemma is the first step towards establishing $\om$-sensibility.
Once again, we write $\cqm\Sigma = (\ptype{} / {\sim},\leq,R,\ell)$, and as usual $R^*$ is the transitive, reflexive closure of $R$.

\begin{lemma}\label{syntactic}
If $\Sigma \subseteq\lanfull$ is finite and closed under subformulas and ${{w}}\in\ptype{} / {\sim}$, then
\begin{enumerate}

\item\label{simulate_next_1} $ \bigvee _{w \mathrel R^* v}\chi^+({{{v}}}) \imp \nx \bigvee _{w \mathrel R^* v}\chi^+(v)   \in \gtl$;

\item\label{simulate_next_2} $ \nx \bigwedge _{w \mathrel R^* v}\chi^-(v)\imp \bigwedge _{w \mathrel R^* v}\chi^-({{{v}}}) \in \gtl$;

\item\label{simulate_yesterday_1} $ \bigvee _{v \mathrel R^* w}\chi^+({{{v}}}) \imp \y\bigvee _{v \mathrel R^* w}\chi^+(v)   \in \gtl$;

\item\label{simulate_yesterday_2} $ \y \bigwedge _{v \mathrel R^* w}\chi^-(v)\imp \bigwedge _{v \mathrel R^* w}\chi^-({{{v}}}) \in \gtl$.

\end{enumerate}
\end{lemma}

\proof
Item~\ref{simulate_next_1} follows from \Cref{propsubplus}(\ref{itPropsubplFive}, as for any $v\in R^*(w)$ we have that
\[ \chi^+({{{v}}}) \imp \nx \bigvee _{v \mathrel R  u}\chi^+(u)  \in \gtl .\]
Since $v \mathrel R  u$ implies that $w \mathrel R^ *  u$ by transitivity,
\[ \chi^+({{{v}}}) \imp \nx \bigvee _{w \mathrel R^*  u}\chi^+(u)  \in \gtl .\]
Since $v$ was arbitrary, we obtain
\[ \bigvee _{w \mathrel R^* v}\chi^+({{{v}}}) \imp \nx \bigvee _{w \mathrel R^* u}\chi^+(u) \in \gtl,  \]
which by a change of variables yields the original claim.

Item~\ref{simulate_next_2} is similar, but uses \Cref{propsubminus}(\ref{itPropsubFive}. Items~\ref{simulate_yesterday_1} and \ref{simulate_yesterday_2} are the temporal duals of items~\ref{simulate_next_1} and \ref{simulate_next_2}.
\endproof

In order to complete our proof that $R$ is $\om$-sensible, it suffices to apply induction to the formulas of \Cref{syntactic}.

\begin{proposition}\label{tempinc}\
\begin{enumerate}

\item\label{ittempincone} If ${{w}}\in\ptype{} / {\sim}$ and $\varphi \until \psi \in \ell ({{w}})$, then there exists ${{{v}}}\in{R^*}({{w}})$ such that $\psi\in \ell ({{{v}}})$.

\item\label{ittempincone2} If ${{w}}\in\ptype{} / {\sim}$ and $\varphi \since \psi \in \ell ({{w}})$, then there exists ${{{v}}}\in(R^{-1})^{*}({{w}})$ such that $\psi\in \ell ({{{v}}})$.

\item\label{ittempinctwo} If ${{w}}\in\ptype{} / {\sim}$ and $\nec \psi\in \Sigma \setminus \ell ({{w}})$, then there exists ${{{v}}}\in{R^*}({{w}})$ such that $\psi\not\in \ell ({{{v}}})$.

\item\label{ittempinctwo2} If ${{w}}\in\ptype{} / {\sim}$ and $\has \psi\in \Sigma \setminus \ell ({{w}})$, then there exists ${{{v}}}\in{(R^{-1})^*}({{w}})$ such that $\psi\not\in \ell ({{{v}}})$.

\end{enumerate}
\end{proposition}

\proof
\noindent\textbf{\ref{ittempincone}.} Towards a contradiction, assume that ${{w}}\in \ptype{} / {\sim}$ and $\varphi \until \psi\in \ell ({{w}})$ but, for all ${{{v}}}\in{R^*}({{w}})$, $\psi \not\in \ell({{v}})$.
By \Cref{syntactic}\eqref{simulate_next_2}, $ \nx (\bigwedge \limits_{w \mathrel R^* v} \chi^-({{{v}}}))\imp \bigwedge\limits_{w \mathrel R^* v} \chi^-({{{v}}}) \in \gtl$.
By the $\ps$-induction axiom \ref{ax12:ind:2} and standard modal reasoning, $ \ps (\bigwedge \limits_{w \mathrel R^* v} \chi^-({{{v}}}))\imp \bigwedge\limits_{w \mathrel R^* v} \chi^-({{{v}}}) \in \gtl$;
in particular,
\begin{equation}\label{other}
 \ps( \bigwedge _{w\mathrel R^* v} \chi^-({{{v}}}))\imp \chi^-({{w}}) \in \gtl.
\end{equation}

Now let ${{{v}}}\in{R^*}({{w}})$.
By \Cref{propsubminus}(\ref{itPropsubOne} and the assumption that $\psi \not\in \ell({{{v}}})$, we have that
$ \psi \imp \chi^-({{{v}}})  \in \gtl$,
and since ${{{v}}}$ was an arbitrary element of ${R^*}({{w}})$, we have
$ \psi \imp \bigwedge_{w\mathrel R^* v}\chi^-({{{v}}})  \in \gtl$.
Using $\nec$-necessitation and the distributivity axiom \ref{ax07:K:Dual} we then have that
$ \ps \psi \imp \ps \bigwedge_{w\mathrel R^* v}\chi^-({{{v}}}) \in \gtl$.
This, along with \eqref{other}, shows that
$ \ps \psi \imp \chi^-({{w}}) \in \gtl$.
Since by \Cref{lemmReverse}\ref{itUtoF}, $\varphi \until \psi \imp \ps \psi$, 
 we can obtain $\varphi \until \psi \imp \chi^-({{w}}) \in \gtl$.
However, by \Cref{propsubminus}(\ref{itPropsubOneb} and our assumption that $\varphi \until \psi\in \ell ({{w}})$ we have that
$ \big ( \varphi \until \psi \imp \chi^-({{w}}) \big ) \imp \chi^-(w) \in \gtl$.
Hence by modus ponens we obtain $ \chi^-({{w}}) \in \gtl$.
Writing $w=[\Gamma]$, \Cref{propSimForm} yields $\chi^-({{w}}) \notin\Gamma^+$, but this contradicts $ \chi^-({{w}}) \in \gtl$.
We conclude that there exists ${{{v}}}\in{R^*}({{w}})$ with $\psi \in \ell({{v}})$, as needed.

\medskip

\noindent\textbf{\ref{ittempincone2}.} This is the temporal dual of item~\ref{ittempincone} (so in particular we use \Cref{syntactic}\eqref{simulate_yesterday_2} in place of \Cref{syntactic}\eqref{simulate_next_2}).

\medskip

\noindent\textbf{\ref{ittempinctwo}.}
This is similar to the $\until$ case, but order dualised.
Towards a contradiction, assume that ${{w}}\in \ptype{} / {\sim}$ and $\nec \psi\in \Sigma \setminus\ell({{w}})$, but for all ${{{v}}}\in{R^*}({{w}})$, $\psi \in \ell({{v}})$.
By \Cref{syntactic}\eqref{simulate_next_1}, $ \bigvee \limits_{w \mathrel R^* v} \chi^+({{{v}}}) \imp \nx \bigvee \limits_{w \mathrel R^* v} \chi^+({{{v}}}) \in \gtl$.
By the $\nec$-induction axiom \ref{ax11:ind:1}, $ \bigvee \limits_{w \mathrel R^* v} \chi^+({{{v}}})\imp\nec \bigvee \limits_{w \mathrel R^* v} \chi^+({{{v}}}) \in \gtl$;
in particular,
\begin{equation}\label{otherb}
\chi^+({{w}}) \imp \nec  \bigvee _{w\mathrel R^*v} \chi^+({{{v}}}) \in \gtl .
\end{equation}

Now let ${{{v}}}\in{R^*}({{w}})$.
By \Cref{propsubplus}(\ref{itPropsubplOneb} and the assumption that $\psi \in \ell({{{v}}})$, we have that
$ \chi^+({{{v}}}) \imp \psi  \in \gtl$,
and since ${{{v}}}$ was arbitrary,
$ \bigvee_{w\mathrel R^* v}\chi^+({{{v}}})\imp \psi  \in \gtl$.
Using the distributivity axiom \ref{ax06KBox} we then have that
$  \nec \bigvee_{w\mathrel R^* v}\chi^+({{{v}}}) \imp \nec \psi \in \gtl$.
This, along with \eqref{otherb}, shows that $\chi^+({{w}}) \imp \nec \psi \in \gtl$. Hence by Rule~\ref{axDimpMon} we obtain
\begin{equation}\label{eqPlusBox}
  (\chi^+({{w}})\dimp\nec \psi) \imp (\nec\psi\dimp\nec\psi) \in \gtl.
\end{equation}
By \Cref{propsubplus}(\ref{itPropsubplOne} and our assumption that $\nec \psi\in \Sigma \setminus\ell({{w}})$ we have that
$ \chi^+(w) \imp \big (  \chi^+({{w}}) \dimp \nec\psi \big ) \in \gtl$. From this together with \eqref{eqPlusBox} we obtain $\chi^+({{w}}) \imp (\nec\psi\dimp\nec\psi) \in \gtl$. In view of  \Cref{lemmReverse}\ref{itNotRef}, this gives $ \chi^+(w) \imp \bot \in \gtl$, so that the pair $ (\chi^+(w), \varnothing)$ is inconsistent.
Writing $w=[\Gamma]$, \Cref{propSimForm} yields $\chi^+({{w}}) \in\Gamma$, contradicting the deduced inconsistency. We conclude that there exists ${{{v}}}\in{R^*}({{w}})$ with $\psi \in\Sigma\setminus \ell({{v}})$.

\medskip

\noindent\textbf{\ref{ittempinctwo2}.} This is the temporal dual of item~\ref{ittempinctwo} (so in particular we use \Cref{syntactic}\eqref{simulate_yesterday_1} in place of \Cref{syntactic}\eqref{simulate_next_1}).
\endproof

\begin{corollary}\label{laststretch}
If $\Sigma \subseteq\lanfull$ is finite and closed under subformulas, then $\cqm\Sigma$ is a quasimodel.
\end{corollary}

\begin{proof}
By \Cref{prop:CisW} and \Cref{prop:quotient_system}, $ \cqm \Sigma$ is a labelled system, while by \Cref{tempinc}, $R$ is $\om$-sensible.
By definition, these facts make $ \cqm \Sigma$ a quasimodel.
\end{proof}

 We are now ready to prove that our calculus is complete.

\begin{theorem}\label{theocomp}
If $\varphi \in \lanfull$ is valid, then $ \varphi \in \gtl$.
\end{theorem}

\proof
We prove the contrapositive.
Suppose $\varphi$ is an unprovable formula and let $\Sigma$ be the set of subformulas of $\varphi$.
Since $\varphi$ is unprovable, there is $\Gamma\in \ptype{}$ with $\varphi\not\in\Gamma$. Then $[\Gamma] \in \ptype{} / {\sim}$ is a point in a quasimodel falsifying $\varphi$, so that by \Cref{second}, $\varphi$ is not valid.
\endproof

\section{PSPACE-complete complexity}\label{Sec:PSPACE}

The proof of \Cref{theorem:decide} yields only a \textsc{nexptime} upper bound for the validity problem.
In this section, we prove that this can be improved to \textsc{pspace}.

First, we recall that the validity problem for $\ltl$ is \textsc{pspace}-com\-plete~\cite[Theorem 4.1]{SK85}. Thus to prove \textsc{pspace}-hardness of the $\gtl$ validity problem, it suffices to give a reduction from $\ltl$ validity to $\gtl$ validity.
Consider the (negative) translation $(\,\cdot\,)^{\bullet}$~\cite{91986055} defined as follows:
 \begin{enumerate}[label=\arabic*)]
 	\item $p^{\bullet} = \neg \neg p$, for each propositional variable $p$, and
 	\item homomorphic for the operators.
\end{enumerate}

In what follows we may assume that $T=\mathbb Z$, equipped with the standard successor function.

To any given $\ltl$ model $(\mathbb{Z},S,V)$,
 we associate a `crisp' G\"odel model $(\mathbb{Z},S,V')$ where $V'(p, t) = 1$ if $t \in V(p)$ and $0$ otherwise. 
 
 \begin{proposition}\label{prop1} For any $\varphi \in \mathcal L$ and for all $t \in \mathbb{Z}$, 
 	\begin{enumerate}[label=\arabic*)]
 		\item if $(\mathbb{Z},S,V), t \models \varphi$ then $V'(\varphi,t) = 1$, and
 		\item if $(\mathbb{Z},S,V), t \not \models \varphi$ then $V'(\varphi,t) = 0$.
 	\end{enumerate}
 \end{proposition}
 
 \begin{proof}
 By structural induction on $\varphi$.
 \end{proof}
 
 Conversely, to any given real-valued G\"odel temporal model $(\mathbb{Z},S,V')$, we associate the crisp G\"odel temporal model $(\mathbb{Z},S,V)$ by fixing $V(p, t) = V'(\neg \neg p, t) \in \lbrace 0,1 \rbrace$. 
 \begin{proposition}\label{prop2} For any $\varphi \in \mathcal L$ and any $t \in \mathbb{Z}$, we have $V'(\varphi^\bullet,t)=V(\varphi,t)$.\end{proposition}

 \begin{proof} By structural induction. 
The case $\varphi\in\mathbb P$ follows from the definitions: $V'(p^\bullet,t)  = V'(\neg \neg p,\allowbreak t)\allowbreak = V(p, t)$. The other cases follow from $({\,\cdot\,})^\bullet$ being homomorphic.
 \end{proof}

 \begin{corollary} For any $\varphi \in \mathcal L$, we have $\ltl \models \varphi$ if and only if ${\gtl} \models \varphi^\bullet$. \end{corollary}

 \begin{proof}
 	For the left-to-right direction, assume by contraposition that $\gtl\not \models \varphi^\bullet$. Therefore there exists a G\"odel temporal model $(\mathbb{Z},S,V')$ and $t \in \mathbb{Z}$ such that $V'(\varphi^\bullet,t) \not= 1$. By \Cref{prop2} there exists a crisp G\"odel temporal model $(\mathbb{Z},S,V)$ such that $V(\varphi,t) \not = 1$. This latter model can be viewed as an $\ltl$ model with $(\mathbb{Z},S,V), t \not \models \varphi$, since the real-valued semantics coincide with classical truth definitions when values are in $\{0,1\}$. Therefore $\ltl\not \models \varphi$.
 	
 	Conversely, assume by contraposition that $\ltl\not \models \varphi$. This means that there exists an $\ltl$ model $(\mathbb{Z},S,V)$ and $t \in \mathbb{Z}$ such that $(\mathbb{Z},S,V), t \not \models \varphi$. Then for the crisp G\"odel temporal model $(\mathbb{Z},S,V')$ as defined in \Cref{prop1} we have $V'(\varphi^\bullet,t) = 0$. As a consequence ${\gtl}\not \models \varphi^\bullet$.
 \end{proof}		 
 
 For the membership of \textsc{pspace}, we adapt the proof of $\ltl$ satisfiability from~\cite[Chapter~6]{degola16a} to the case of $\gtl$.
 Define an \define{ultimately periodic quasimodel} to be a quasimodel $\mathcal Q=(W ,{\leq},\ell,S)$ such that
 there is a `double lasso' graph $(T,R)$ with $T=\{-(i'+l'-1),\dots,i+l-1\}$, $ k \mathrel R (k+1)$ for $-(i'+l'-1) \leq k \leq i+l-2$, and also $(i+l)\mathrel R i$ and $-i'\mathrel R -(i'+l')$, and a `projection' function $\pi\colon W\to T$ such that each $\pi^{-1}(t)$ is an entire linear component of $W$, and $w \mathrel S v$ implies $\pi(w)\mathrel R\pi(v) $.

In other words, $\mathcal Q$ has an underlying nondeterministic `flow' $T$ consisting of a loop, followed by a middle segment (containing $0$), and finally a second loop, and each $t\in T$ is assigned a linear order $\pi^{-1}(t)$, which we may also write as $W_t$.
  
 \begin{theorem}[ultimately periodic quasimodel property]\label{aperiodic} Every falsifiable $\mathcal L$-formula is falsifiable in an ultimately periodic quasimodel of height bounded by $|\Sigma|+1$.
 \end{theorem}

 \begin{proof}
 We sketch the construction.
 By \Cref{quasi:bound}, if $\varphi$ is falsifiable, it is falsifiable on some quasimodel $\mathcal Q' = (W',{\leq}'  ,\ell' ,S' )$ of height at most $\lgt \Sigma+1$.
 Choose $w_0 \in W' $ such that $\varphi\notin\ell_0(w_0)$, and let $W_0$ be the linear component of $w_0$ (i.e.~$W_0 = \{v\in W \mid v\leq w_0\text{ or }w_0\leq v\}$) and $\leq_0$ be the restriction of $\leq'$ to $W_0$.
 By a priority method similar to that of \Cref{SecGen}, we define a bi-infinite sequence $\dots, (W_{-1},{\leq}_{-1}),(W_0,{\leq}_0),(W_1,{\leq}_1),\dots$ and sensible relations $S_k\subseteq W_k\times W_{k+1}$, such that $\mathcal Q^{\infty} = (W^\infty,
 {\leq}^\infty,\ell^\infty,S^\infty) $ is a quasimodel, where $ W^\infty =\bigsqcup_{k\in\mathbb Z} W_k$ ($\bigsqcup$ denotes a disjoint union),  $ {\leq ^\infty} = {\bigsqcup_{k\in\mathbb Z} \leq_k}$, and so on.

 Note that there are at most $2^{\lgt \Sigma(\lgt\Sigma+1)}$ possible choices of $W_k$, since each $W_k$ consists of at most $\lgt\Sigma+1 $ types, and there are at most $2^{\lgt\Sigma}$ types.
 This in particular implies that some $W_i$ repeats infinitely often for arbitrarily large values of $i$.
 Let $l>0$ be such that $W_{i+l } = W_i$ and every $\until$ or $\nec$ defect of $W_i$ has been realised before $W_{i+l}$; such an $l$ exists because $W_i$ has finitely many defects.
 Similarly, some $W_{-i'}$ repeats for infinitely many values of $i'$, and we choose $l'>0$ such that $W_{-(i'+l ') } = W_{-i'}$ and all $\since$ or $\has$ defects of $W_{-i'}$ have been realised after $W_{-(i'+l')}$.
 We define $\mathcal Q=(W,{\leq},\ell,S)$ to be the restriction of $\mathcal Q^\infty$ to $\bigcup_{k=-(i'+l'-1)}^{i+l-1}W_k$, but with $S$ redefined on $W_{i+l-1 }$ so that it maps to $W_i$ and redefined on $W_{-(i'+l' -1)}$ so that it maps from $W_{i'}$.

It remains to check that $\mathcal Q$ is a quasimodel.
We only check that it is $\om$-sensible, as the other properties are easy to check.
Consider the case of $\varphi \until \psi\in \ell(w)$ (the cases for other temporal modalities are similar).
Then $w\in W_k$ for some $k$, which means that for some $j$ (namely, $j=i+l-k$), there is $v\in W_i$ such that $w\mathrel S^j v$.
If $\varphi\until \psi\notin \ell(v)$, then this defect must already have been realised.
Otherwise, by construction, there are some $j'$ and some $u$ such that $v\mathrel S^{j'} u$ and $\psi\in \ell(u)$.
Thus $w \mathrel S^{j+j'}  u   $ and $\psi \in \ell(u)$, as required.
 \end{proof}

Ultimately periodic models can be represented using sets of \emph{moments}.

\begin{definition}\label{definition:moments}
A \define{$\Sigma$-moment} is a sequence of the form $\mathfrak m = (\mathfrak m_0,\dots,\mathfrak m_{m} )$, where 
\begin{enumerate}
\item
each $\mathfrak m_i$ is a $\Sigma$-type, 
\item\label{two}
$\mathfrak m_i\supsetneq \mathfrak m_{i+1}$ for $i<m-1$, 
\item\label{three}
 for every $\varphi \imp \psi \in \Sigma\setminus \mathfrak m_{i} $ there is some $j\leq i$ with $\varphi \in \mathfrak m_{j}$ but $\psi \not\in \mathfrak m_{j}$,
 \item\label{four}
 for every $\varphi \dimp \psi \in  \mathfrak m_{i} $ there is some $j\geq i$ with $\varphi \in \mathfrak m_{j}$ but $\psi \not\in \mathfrak m_{j}$.
 \end{enumerate}
 We write $|\mathfrak M|$ for the set $\{\mathfrak m_0,\dots,\mathfrak m_m\}$.
The set of $\Sigma$-moments is denoted $M_\Sigma$.
\end{definition}


We define the labelled space $(\mathfrak m_0,\dots,\mathfrak m_{m} ) + (\mathfrak n_0,\dots,\mathfrak n_{n} )$ to be the parallel sum of the two linear posets $(\mathfrak m_0,\dots,\mathfrak m_{m} )$ and $(\mathfrak n_0,\dots,\mathfrak n_{n})$ with labelling given by the identity.

\begin{definition}\label{ts}
The moment $\mathfrak n$ is a \define{temporal successor} of $\mathfrak m$, denoted $\mathfrak m \mathrel S_\Sigma \mathfrak n$, if there exists a fully confluent convex sensible relation
$R\subseteq |\mathfrak m| \times |\mathfrak n|$ on the labelled space $\mathfrak m + \mathfrak n$.
\end{definition}


\begin{definition}
 We define
$\moments\Sigma=\<M_\Sigma, S_\Sigma\>$.
\end{definition}





Because of condition \eqref{two} in \Cref{definition:moments}, if $\Sigma$ is finite then so is $\moments\Sigma$.

 \begin{definition} A \define{small falsifiability 
  witness} for an $\mathcal L$-formu\-la $\varphi$ is a finite sequence of $\Sigma$-moments $\mathfrak{m}^{-i'-l'},\dots,\mathfrak{m}^{- i' },\dots, \mathfrak{m}^i, \dots, \mathfrak{m}^{i+l}$ with distinguished positions $-i',i$ and binary relations $\emptyset \neq S_k \subseteq |\mathfrak m^k| \times |\mathfrak m^{k+1}|$ for each $k$ with $-i'-l'\leq k<i+l$ such that 
 	\begin{enumerate}[label=(\Alph*)]

 		\item $\varphi \not\in \Phi$ for some $\Phi \in \mathfrak m^0$
 		
 		\item
 		$\mathfrak{m}^{-i'} = \mathfrak{m}^{-i'-l'}$ and $\mathfrak{m}^i = \mathfrak{m}^{i+l}$,

 		\item each $S_k$ is fully confluent, convex, and sensible,

 		\item if $\varphi\until \psi \in \mathfrak m^i_j$ then there are $r<l$ and a sequence $(j_k)_{k\leq r}$ with $j_0=j$ such that $\mathfrak m^{i+k}_{j_k} \mathrel S_{i+k} \mathfrak m^{i+k+1}_{j_{k+1}} $ and $\varphi \in \mathfrak m^{i+k}_{j_{k}}$ for all $k<r$, and $\psi\in \mathfrak m^{i+r}_{j_{r}}$,
 		
 		\item if $\nec \psi \in\Sigma\setminus \mathfrak m^i_j$ then there are $r<l$ and a sequence $(j_k)_{k\leq r}$ with $j_0=j$ such that $\mathfrak m^{i+k}_{j_k} \mathrel S_{i+k} \mathfrak m^{i+k+1}_{j_{k+1}} $ for all $k<r$ and $\psi\notin \mathfrak m^{i+r}_{j_{r}}$,
 		
 		 		\item if $\varphi\since \psi \in \mathfrak m^{i'}_j$ then there are $r<l'$ and a sequence $(j_k)_{k\leq r}$ with $j_0=j$ such that $\mathfrak m^{-i'-k-1}_{j_{k+1}} \mathrel S_{-i'-k} \mathfrak m^{-i'-k }_{j_{k}} $ and $\varphi \in \mathfrak m^{-i'-k}_{j_{k}}$ for all $k<r$, and $\psi\in \mathfrak m^{-(i'+r)}_{j_{r}}$,
 		
 		\item  if $\has \psi \in\Sigma\setminus \mathfrak m^{-i'}_j$ then there are $r<l'$ and a sequence $(j_k)_{k\leq r}$ with $j_0=j$ such that $\mathfrak m^{-i'-k-1}_{j_{k+1}} \mathrel S_{-i'-k} \mathfrak m^{-i'-k }_{j_{k}} $ for all $k<r$, and $\psi\notin \mathfrak m^{-i'-r}_{j_{r}}$.
 				
 	\end{enumerate}
 \end{definition}

\Cref{aperiodic} implies that if an $\mathcal L$-formula is falsifiable then it has a small falsifiability witness. Moreover, the converse is also true, because the small falsifiability witness can be viewed as a quasimodel. As a consequence, we obtain an equivalence between the existence of a possibly infinite structure (a model of $\varphi$) and the existence of a finite structure (a small falsifiability witness) for a given $\mathcal L$-formula $\varphi$.
 
 \begin{theorem}	\label{witness}
 	An $\mathcal L$-formula is falsifiable if and only if it has a small falsifiability witness.	
 \end{theorem}

 \begin{proof} 	
 		For the left-to-right direction, assume that the formula $\varphi$ is falsifiable.
 		By \Cref{aperiodic}, there exists an ultimately periodic quasimodel $\mathcal{Q} = (W,{\leq},\ell,S)$ such that $W=\bigcup_{k=-i'-l'}^{i+l} W_k$ and $\varphi \not\in \ell(w)$ for some $w\in W_0$.
 		For each $W_k$, let $W_k = \{v^k_0,\dots, v^k_{m_k}\}$ in increasing order, and let $\mathfrak m^k = ( \ell (v^k_0), \dots, \ell (v^k_{m_k}) )$ (deleting repeating types if needed).   	
 		It is easy to check that the sequence $\mathfrak m^{-i'-l'},\dots, \mathfrak m^{i'}, \dots, \mathfrak{m}^0,\dots,\mathfrak{m}^i, \dots, \mathfrak{m}^{i+l}$, with sensible relations $S_k$ defined in the obvious way, yields a small falsifiability witness.
 		 	
 	Conversely, we will show that if a formula has a small falsifiability witness $\mathfrak{m}^{-i'-l'}, \dots,\allowbreak \mathfrak{m}^{- i' }, \dots, \mathfrak{m}^i, \dots, \mathfrak{m}^{i+l}$ then it is falsifiable.
 	Write $\mathfrak{m}^{k} = (\mathfrak{m}^{k}_0,\dots,\mathfrak{m}^{k}_{m_k})$ and consider the labelled space $\mathcal{Q} = (W,{\leq}, \ell, S)$, where $W= \{ (\mathfrak m_s^k,k) \mid -(i'+l')<k<i+l\text{ and } 0 \leq s\leq m_k \}$ and ${\leq}$, $\ell$, and $S$ are defined in the obvious way. 
It is not hard to check that $\mathcal{Q} $ is a quasimodel falsifying $\varphi$. Hence by \Cref{sound}, $\varphi$ is falsifiable.
 \end{proof}	


 \begin{algorithm}[h!]
 \DontPrintSemicolon
 	\caption{$\gtl$ falsifiability algorithm}\label{gtlsat}
 	 input $\varphi$\;
 	 set $\Sigma$ to be the set of subformulas of $\varphi$\;
 	 guess three moments, $\mathfrak{m}^p$, $\mathfrak{m}$, and $\mathfrak m^f$, of heights $s,m,n\leq |\Sigma|+1$ such that $\varphi\notin \mathfrak m_m$\;
	 $\mathfrak{m}' \gets \mathfrak{m}$\;

	\While{$\mathfrak m\neq \mathfrak m^f$ \textup{or} $\mathfrak m'\neq \mathfrak m^p$}{
 		guess a moment $\mathfrak{n}$ of height at most $|\Sigma|+1$\;
	 	\uIf{$\mathfrak m\neq \mathfrak m^f$ \textup{and} $\mathfrak n$ \textup{is a temporal successor of} $\mathfrak m$}{$\mathfrak{m} \leftarrow \mathfrak{n}$}
		\uElseIf{$\mathfrak m'\neq \mathfrak m^p$ \textup{and} $\mathfrak n$ \textup{is a temporal predecessor of} $\mathfrak m'$}{ $\mathfrak{m}' \leftarrow \mathfrak{n}$}
 		\Else{reject}
 	}
\Comment{Initialise defects}
$\Delta \leftarrow {\{ (k,\ps \psi) \mid \varphi \until \psi \in  \mathfrak{m}_k,\; 0 \le k \leq n \} \cup \{ (k,\past \psi) \mid \varphi \since \psi \in  \mathfrak{m}'_k,\; 0 \leq k \leq s \}} \allowbreak{\cup \{ (k,\nec \psi) \mid \nec \psi \in \Sigma\setminus {\mathfrak{m}}_k,\; 0 \le k \leq n \} \cup \{ (k,\has \psi) \mid \has \psi \in \Sigma\setminus {\mathfrak{m}'}_k,\; 0 \le k \leq s \}}$\;
\Comment{Initialise cured defects}
$\Gamma \leftarrow {\lbrace (k,\ps \psi) \mid \psi \in  \mathfrak{m}_k,\; 0 \le k \leq n \rbrace \cup \lbrace (k,\past \psi) \mid \psi \in  \mathfrak{m}'_k,\; 0 \le k \leq s \rbrace}\allowbreak{ \cup \lbrace (k,\nec \psi) \mid \psi \in \Sigma\setminus{\mathfrak{m}}_k,\; 0 \le k \leq n \rbrace \cup\lbrace (k,\has \psi) \mid \psi \in \Sigma\setminus \mathfrak{m}'_k,\; 0 \le k \leq s \rbrace }$\;

\Comment{Initialise reachability relations}             
$S^* \leftarrow \lbrace (k,k)  \mid 0 \le k \leq n \rbrace$\;
$(S^{-1})^* \leftarrow \lbrace (k,k)  \mid 0 \le k \leq s \rbrace$\;

\While{$\mathfrak{m}^f \neq \mathfrak{m}$ \textup{or} $\mathfrak{m}^p \neq \mathfrak{m}'$ \textup{or} $\Delta \not \subseteq \Gamma$}{
	guess a moment $\mathfrak{n}$ of height at most $|\Sigma|+1$\;
 	\uIf{$\mathfrak m\neq \mathfrak m^f$ \textup{and} $\mathfrak n$ \textup{is a temporal successor of} $\mathfrak m$, \textup{witnessed by $R$}}{
 		$S^* \leftarrow \lbrace (k,z) \mid \exists y : (k,y)  \in S^* \hbox{ and } (\mathfrak{m}_y,\mathfrak{n}_z) \in R\rbrace$\;
		$\Gamma \leftarrow \Gamma \cup \lbrace  (k,\ps \psi) \mid \exists z : (k,z) \in S^* \hbox{ and } \psi \in  \mathfrak{n}_z\rbrace$\;
		$\Gamma \leftarrow \Gamma \cup \lbrace  (k,\nec \psi) \mid \exists z : (k,z) \in S^* \hbox{ and } \psi \in \Sigma \setminus {\mathfrak{n}_z}\rbrace$\;
		$\mathfrak{m} \leftarrow \mathfrak{n}$\;

 	}
	\uElseIf{$\mathfrak m'\neq \mathfrak m_p$ \textup{and} $\mathfrak n$ \textup{is a temporal predecessor of} $\mathfrak m'$, \textup{witnessed by $R$}}{

	$(S^{-1})^* \leftarrow \lbrace (z,k) \mid \exists y : (y,k)  \in (S^{-1})^* \hbox{ and } (\mathfrak{m}'_z,\mathfrak{n}_y) \in R^{-1}\rbrace$\;                     
	$\Gamma \leftarrow \Gamma \cup \lbrace  (k,\past \psi) \mid \exists z : (k,z) \in (S^{-1})^* \hbox{ and } \psi \in  \mathfrak{n}_z\rbrace$\;
	$\Gamma \leftarrow \Gamma \cup \lbrace  (k,\has \psi) \mid \exists z : (k,z) \in (S^{-1})^* \hbox{ and } \psi \in \Sigma \setminus {\mathfrak{n}_z}\rbrace$\;
	$\mathfrak{m}' \leftarrow \mathfrak{n}$\;	
	}
	\Else{reject}         
}
accept

 \end{algorithm}

 \afterpage{\clearpage}
 
 From \Cref{witness} we may define Algorithm~\ref{gtlsat}, which nondeterministically checks for falsifiability in \textsc{pspace}.
 According to Savitch's theorem~\cite{SAVITCH1970177}, nondeterministic polynomial space is equal to deterministic polynomial space; thus we obtain the following.
 
 \begin{theorem}
 There exists a deterministic algorithm for falsifiability checking of an $\mathcal L$-formula that is correct and works in space that is polynomial in the size of the input formula.
 \end{theorem}

\begin{proof} We argue that Algorithm~\ref{gtlsat} is a correct and complete nondeterministic polynomial-space algorithm for falsifiability. 
	If $\varphi$ is falsifiable, then by \Cref{witness}, $\varphi$ has a small falsifiability witness $\mathfrak m^{-i'-l'},\dots,\mathfrak m^{-i'} ,\dots, \mathfrak{m}^0 ,\dots, \mathfrak{m}^i ,\dots, \mathfrak{m}^{i + l}$. 
	We initialise both $\mathfrak m$ and $\mathfrak{m}'$ to $\mathfrak m^0$. With the first while loop, we initialise $\mathfrak m^s$ to $\mathfrak m^{i'}$, initialise  $\mathfrak m^f$ to $\mathfrak m^i$. During each iteration of the second while loop, we choose $\mathfrak n$ to be either the successor of $\mathfrak m$ or the predecessor of $\mathfrak{m}'$. 
    This yields an accepting computation of Algorithm~\ref{gtlsat}. In particular, since $\mathfrak m^{-i'-l'},\dots,\mathfrak m^{-i'} ,\dots, \mathfrak{m}^0 ,\dots, \mathfrak{m}^i ,\dots, \mathfrak{m}^{i + l}$ has no defects, for each element of $\Delta$, at some stage a `cure' is found witnessing it is not in fact a defect. That is, we eventually obtain $\Delta \subseteq \Gamma$.
	
	Conversely, if Algorithm~\ref{gtlsat} has an accepting computation, let $\mathfrak m^{-i'-l'},\dots,\mathfrak m^{-i'}, \dots, \mathfrak{m}^0, \allowbreak\dots,\allowbreak \mathfrak{m}^i, \dots, \mathfrak{m}^{i + l}$ enumerate the values taken by $\mathfrak{m}$ and $\mathfrak{m}'$ where $i'$ and $i$ are the least natural numbers such that $\mathfrak m^{i'}=\mathfrak m^s$ and $\mathfrak{m}^i = \mathfrak{m}^f$, respectively.
	It is not hard to check that this sequence yields a small falsifiability witness.
	
To check that the nondeterministic algorithm uses polynomial space, it is sufficient to observe that each subset of $\Sigma$ can be encoded by a polynomial number of bits.
	Since $\mathfrak{m}$, $\mathfrak{m}'$, $\mathfrak{m}^f$, $\mathfrak{m}^s$ and $\mathfrak{n}$ have height at most $\lvert \Sigma \rvert + 1$ we need $7 \lvert \Sigma \rvert + 7$ of those sets (at most $\lvert \Sigma \rvert + 1$ for each of $\mathfrak{m}$, $\mathfrak{m}'$, $\mathfrak{m}^s$, $\mathfrak{m}^f$, $\mathfrak{n}$, $\Delta$, and $\Gamma$).
	Checking that $\mathfrak n$ is a temporal successor (respectively predecessor) of $\mathfrak m$  (respectively $\mathfrak{m}'$) can be done nondeterministically by guessing a relation $R$ and checking that it is a sensible and bi-serial relation; but the size of $R$ is bounded by the product of the sizes of $\mathfrak m$ (respectively $\mathfrak{m}'$) and $\mathfrak n$.
	Similarly, both $S^*$ and $ (S^{-1})^* $ have at most $|\Sigma|^2$ elements, so also only require polynomial space.
\end{proof}

  With this we conclude that the validity problem is \textsc{pspace}-complete. 

 \begin{theorem}\label{thmPSPACE}
 	The validity problem for ${\gtl}$ is \textsc{pspace}-complete.
 \end{theorem}

 \section{Concluding remarks}\label{SecConc}

 We have defined a natural version of \emph{linear temporal logic} with a G\"odel--Dummett base, suitable for reasoning with vague or incomplete information, and shown that it may equivalently be characterised as a fuzzy logic or as a superintuitionistic logic using standard semantics in each case.
 Despite the lack of a finite model property for either of the two semantics, we have introduced a class of \emph{quasimodels} for which $\gtl$ \emph{does} satisfy a version of the finite model property, and moreover shown how these quasimodels can be used to adapt the classical proof that the validity problem is \textsc{pspace}-complete.

This \textsc{pspace}-complete complexity puts G\"odel temporal logic in sharp contrast to other fuzzy logics, whose transitive modal versions are undecidable~\cite{Vidal21}, or intuitionistic temporal logics, where systems are known to be decidable only with non-elementary upper bounds, if at all \cite{F-D18,BalbianiToCL}.
This places G\"odel--Dummett logic as \emph{the} premier base for computational applications of sub-classical modal and temporal logics, as far as complexity is concerned.

We have also provided a sound and complete calculus for the G\"odel temporal logic $\gtl$.
This result further cements $\gtl$ as a privileged logic for temporal reasoning with non-binary degrees of truth and paves the way for a proof-theoretic treatment of this logic.
Among the challenges in this direction is the design of cut-free or cyclic calculi.

In proving our main results, we have developed tools for the treatment of superintuitionistic temporal logics, specifically identifying the usefulness of  combining the presence of `henceforth' with \emph{co-implication}.
We believe that this insight will lead to completeness proofs for related logics, including intuitionistic $\ltl$, where complete calculi for `eventually' are available, but not for `henceforth'.
Regarding this, it should be remarked that the techniques of \cite{eventually} should lead to a sound and complete calculus for the logic with implication, `next', and `eventually' (but no co-implication or henceforth), although such a result does not follow immediately from the present work.
The techniques we have used to prove completeness and decidability are quite robust and we expect that they may be applicable to more expressive logics such as the linear-time $\mu$-calculus (see \cite{10.1145/73560.73582}) or the branching-time logics $\mathsf{PDL}$, $\mathsf{CTL}^*$, or even the full modal $\mu$-calculus.
This represents a milestone in the program pioneered by Caicedo et al.~\cite{CaicedoMRR17,CaicedoR10} of extending results from classical modal and temporal logics to their G\"odel counterparts.

Another subject that would be worth studying in the near future is  \emph{bisimulation} in G\"odel temporal logic. This tool has been used to determine that temporal operators are not interdefinable in the intuitionistic temporal setting~\cite{Balbiani2017,BalbianiToCL}.
For the class of temporal here-and-there models, `henceforth' is a basic operator that cannot be defined, while `eventually' becomes definable in terms of `henceforth', `next', and implication~\cite{Balbiani2017,BalbianiToCL}. When including co-implication, results on definability exist in the literature: for a combination of the logic of here-and-there, co-implication and the basic modal logic $\mathsf K$, it was proven in \cite{BD18} that the `possibly' and `necessarily' modalities become interdefinable.
We do not know if co-implication has a similar effect on our G\"odel temporal logic; a negative answer would require a suitable notion of bisimulation preserving both implications as well as the temporal operators. 

\subsection*{Acknowledgments} 
This work has been partially suppor\-ted by FWO--FWF grant G030620N/\allowbreak I4513N (Aguilera and Fern\'andez-Duque), FWO grant 3E017319 (Aguilera), the projects EL4HC and \'etoiles montantes CTASP at R\'egion Pays de la Loire, France (Di\'eguez), the COST action CA-17124 (Di\'eguez and Fern\'andez-Duque), and SNSF--FWO Lead Agen\-cy Grant 200021L\_196176/\allowbreak G0E2121N (McLean and Fern\'andez-Duque).



%
%
%

\end{document}